%% file: main_arxiv.tex
\documentclass{article}

\usepackage[creativeai, preprint, nonatbib]{neurips_2026}
\usepackage[numbers,compress]{natbib}

\usepackage[utf8]{inputenc}
\usepackage[T1]{fontenc}
\usepackage{lineno}  % needed for \nolinenumbers wrappers in 03_theory_v2.tex
\renewcommand{\linenumbers}{}  % arXiv: keep line numbers OFF (neutralize the \linenumbers call in 03_theory_v2)

\PassOptionsToPackage{hyphens}{url}
\usepackage{hyperref}

\usepackage{url}
\usepackage{booktabs}
\usepackage{amsfonts}
\usepackage{amsmath}
\usepackage{amssymb}
\usepackage{amsthm}
\usepackage{mathtools}
\usepackage[capitalise]{cleveref}
\usepackage{microtype}
\usepackage{xcolor}
\usepackage{graphicx}
\usepackage{longtable}
\usepackage{float}
\usepackage{placeins}

\newtheorem{theorem}{Theorem}
\newtheorem{proposition}{Proposition}

\newtheorem{corollary}{Corollary}
\newtheorem{lemma}{Lemma}
\newtheorem{remark}{Remark}

\DeclareMathOperator{\Var}{Var}
\DeclareMathOperator{\Cov}{Cov}
\DeclareMathOperator{\SNR}{SNR}

\newcommand{\E}{\mathbb{E}}
\newcommand{\R}{\mathbb{R}}

\newcommand{\lz}[1]{}
\newcommand{\cd}[1]{}
\newcommand{\jd}[1]{}
\newcommand{\jm}[1]{}
\newcommand{\bl}[1]{}

\title{What's a Credit Worth? A Market Framework for Attribution-Aware Compensation in Generative Music}

\author{
  Luyang Zhang$^{1}$\thanks{Equal contribution.} \quad Xirui Jiang$^{2}$\footnotemark[1] \quad Junwei Deng$^3$ \\
  \textbf{Beibei Li}$^1$ \quad \textbf{Jiaqi W. Ma}$^3$ \quad \textbf{Chris Donahue}$^1$ \\[0.5em]
  $^1$Carnegie Mellon University \quad $^2$University of Michigan, Ann Arbor \\
  $^3$University of Illinois Urbana-Champaign \\[0.5em]
  \texttt{\{luyangz, beibeili\}@andrew.cmu.edu} \\
  \texttt{xirui@umich.edu} \\
  \texttt{\{junweid2, jiaqima\}@illinois.edu}
}

\begin{document}

\maketitle

\input{section_arxiv/00_abstract}

%----------------------------------------------------------------------
% Body
%----------------------------------------------------------------------
\input{section_arxiv/01_intro}
\input{section_arxiv/02_related}
\input{section_arxiv/03_theory_v2}
\input{section_arxiv/05_empirical}
\input{section_arxiv/06_conclusion}

%----------------------------------------------------------------------
% References
%----------------------------------------------------------------------
\bibliographystyle{plainnat}
\bibliography{refs}

%----------------------------------------------------------------------
% Appendix
%----------------------------------------------------------------------
\clearpage
\appendix
\crefalias{section}{appendix}
\crefalias{subsection}{appendix}
\input{section_arxiv/A_gaussian}
\input{section_arxiv/B_proofs}
\input{section_arxiv/C_ic}
\input{section_arxiv/D_uniform_welfare}
\input{section_arxiv/D2_competition}
\input{section_arxiv/E_modeling}
\input{section_arxiv/F_additional_empirical}
\input{section_arxiv/G_bargaining}
\input{section_arxiv/H_dynamics}
\input{section_arxiv/I_synthetic}

\end{document}

%% file: section_arxiv/00_abstract.tex
% !TEX root = ../main.tex
\begin{abstract}
Advances in generative AI are rapidly increasing the quality and commercial value of generated music, and this progress depends on large catalogs of creators' recordings.
This raises a central question for platform design: how should creators be compensated when their work is used to train generative AI models that in turn produce commercial outputs?
We develop a framework for fairly compensating creators in generative-music markets, where each creator's payment depends on a data-attribution score estimating their contribution to model outputs. 
% CHRIS: We might want something like this to distinguish our payment framework from others that have been proposed?
Compared to past compensation frameworks, our framework has two unique considerations: (1) attribution is traced to entire creator catalogs, 
%, 
not individual songs, and (2) the \emph{informativeness} (signal-to-noise ratio) of the attribution score is an input to the payment mechanism.
The framework yields a closed-form payment rule per creator and measures the welfare cost of inaccurate attribution for both creators and the platform.
Whether the welfare-optimal contract is royalty-based or takes the form of fixed-fee licensing depends on how informative attribution is for that creator's catalog.
We show that better attribution translates directly into welfare gains for both creators and the platform, yet under multi-platform competition a platform only captures gains from attribution improvements when its signal becomes the most precise in the market. 
%Empirically, we propose leave-one-catalog-out fine-tuning as a scalable measurement protocol for catalog-level attribution, applied to a generative music model.
To ground our framework in empirical behavior, 
we train acoustic and symbolic music generation models and measure the informativeness of scalable attribution techniques against a leave-one-catalog-out ground truth. 
Our experiments reveal that noisy attribution signals push payment toward fixed-fee licensing and diminish welfare for both creators and the platform, providing an economic motivation for further research on improved attribution.
% CHRIS: below sentence felt a bit redundant (felt like we made this point several times) so I shortened
%This calls for improving attribution signal quality, which raises welfare on both the platform and creators.
\end{abstract}

%% file: section_arxiv/01_intro.tex
% !TEX root = ../main.tex
\section{Introduction}
\label{sec:intro}
%----------------------------------------------------------------------
Generative music models trained on copyrighted recordings are reshaping how music is produced and consumed. As their output capabilities have advanced, 
AI music is crowding music streaming services~\citep{deezer2026aiuploads},
copyright disputes have emerged~\citep{riaa2024lawsuit}, 
and regulators have responded with disclosure mandates~\citep{euaiact2024}, suggesting the absence of a reliable and fair compensation rule. Training-data provenance is rarely disclosed, and licensing, where it exists, is negotiated at the catalog or genre level for flat fees that ignore two key factors: 
%which creators contributed most to a given output \citep{ifpi2025}
(1)~that some data contributes more than others for a given generated output~\citep{ifpi2025}, and 
(2)~that flat licensing fees do not grow proportionally to revenue from generated outputs. 
This raises a central question: how should creators be compensated when their work trains generative AI models that are then used commercially?

Answering this question requires combining two technical fields. The first is \emph{data attribution}, which estimates how much each creator's catalog contributed to a given model output. The second is \emph{payment design}, which turns those estimates into contracts that handle noise, risk, and creator participation.
For data attribution, machine-learning methods \citep{ghorbani2019data,koh2017understanding,ilyas2022datamodels} can estimate each creator's contribution, but the estimates carry noise that varies by creator and by method. 
Past work in data attribution for music AI~\citep{barnett2024embeddings,choi2025unlearning,deng2023copyright} seek to measure or reduce this noise, but do not factor it into the payment mechanism. 
For payment design, contract theory \citep{holmstrom1979moral} supplies general tools for payment under noisy signals, but has rarely been combined with the attribution methods that produce the signal.
We unify these isolated literatures.

% P3 — what our framework is about
To this end, we develop a framework for fairly compensating creators in generative-music markets, where each creator's payment depends on a data-attribution score estimating their contribution to model outputs. 
The framework's central concept is \emph{informativeness}, a measure of how well the noisy attribution signal tracks the true contribution. 
While most work in data attribution tracks influence per-datum (i.e.,~per song), we make a simplifying assumption that tracking influence \emph{per-catalog} may suffice for compensation, as creators' catalogs naturally group the training data by distinct payee.
In practice, here \emph{creator} can refer to any type of music rights holder, e.g., an independent aritst, or a record label. 
We derive a closed-form payment rule per creator that balances a fixed \emph{licensing} fee and \emph{royalties} proportional to the catalog attribution score,
% CHRIS: this feels redundant w.r.t. the "implications" discussed int he next paragraph
% (more royalty when informativeness is high, more fixed fee when it is low), 
along with a closed-form expression for the welfare loss caused by imperfect 
%attribution. 
informativeness. 

The framework yields several implications about how attribution informativeness shapes contracts, welfare, and creator incentives.
For music AI platforms, it is optimal to adopt royalty-based payment when the attribution to a creator's catalog is precise, and fixed-fee licensing when it is noisy.
\textbf{We show that both creators and the AI platform can benefit from more precise attribution.}
Under multi-platform competition, the leading platform is mostly incentivized to improve the attribution quality, since better signals widen its lead and let it capture more value from creators.
For creators, the welfare-optimal contract form (royalty-based or fixed-fee licensing) is determined by the platform's attribution informativeness for their catalog. When their own beliefs about that informativeness diverge from the platform's measurement, both creators and the platform incur welfare losses.%  
% CD: Why does this produce systemic under-investment? Not obvious to me --- seems like it could incentive investment depending on the magnitude of potential benefits. Also, I thought we were only studying single platform scenarios? I'll keep reading...
% CD: I don't follow this argument. How would creators report attribution scores? Doesn't the platform have to report attribution scores?
% LZ: reshape the positioning based on the overall feedback

% P5 — empirical
Under this framework, we propose a leave-one-catalog-out measurement pipeline and study it empirically across two settings.
The first setting applies a latent audio diffusion model Stable Audio 
Open \citep{evans2024stable} on $22$ catalogs of Creative Commons music data \citep{bogdanov2019mtg} and measure informativeness of EKFAC \citep{mlodozeniec2025influence} and D-TRAK \citep{zheng2024dtrak}.
The second setting (symbolic generation) trains a music language model~\citep{thickstun2023anticipatory} on $100$ catalogs of commercial melodies and harmonies from TheoryTab \citep{donahue2022melody} and evaluate LoGra \citep{choe2024your} and TRAK \citep{park2023trak}.
We observe that two attribution methods can \emph{rank} creators similarly yet differ by orders of magnitude in the \emph{scores} that determine payment.
In the audio setting, current attribution noise makes fixed-fee licensing welfare-optimal for nearly all creators.
Creators with highly identifiable styles would benefit from royalty-based payment, but current audio attribution is too noisy to track their contribution.
The symbolic setting reproduces the same diagnostic --- current attribution methods sit below the royalty threshold across both modalities --- confirming that the precision shortfall is a property of current ML broadly, not a quirk of audio diffusion.
Platforms deploying attribution-based compensation should therefore benchmark multiple methods rather than commit to one, benefiting both creators and the platform.
%Stable Audio Open, a commercial-grade audio diffusion model, with EKFAC \citep{mlodozeniec2025influence} and D-TRAK \citep{zheng2024dtrak}. 
%uses a symbolic-music language model for higher-precision attribution.

% P6 — three contributions
Our contributions are threefold:
\begin{itemize}
\item \textbf{Framework.}
The first principal-agent framework bridging data attribution and contract design for generative music markets. 
It yields closed-form expressions for the optimal payment rule per creator (royalty-based or fixed-fee licensing) and for the welfare cost of imperfect attribution.

\item \textbf{Measurement.}
The first empirical benchmark of attribution informativeness in 
%deployed 
generative music models, with experiments that reflect the model paradigms and content modalities of today's commercial music AI systems under tractable attribution methods. %

\item \textbf{Implications.}
Combined theoretical and empirical analysis yields concrete 
% CHRIS: just trying to get rid of dangler. %economic 
implications for the generative music industry, establishing attribution quality and market transparency as measurable economic primitives and providing a basis for industry policy and governance.
\end{itemize}

%----------------------------------------------------------------------

%% file: section_arxiv/02_related.tex
% !TEX root = ../main.tex
\section{Related work}
\label{sec:related}
%----------------------------------------------------------------------

\textbf{Data pricing and contract design.}
Several works study the design of data markets---marketplace design \citep{agarwal2019marketplace}, the price of information \citep{bergemann2018design}, and pricing data for an ML buyer \citep{chen2022selling}---alongside macroeconomic analyses of data as a nonrival good \citep{jones2020nonrivalry,cong2022data,acemoglu2022toomuch}. Contract theory provides the payment-design tools: \citet{holmstrom1979moral} characterizes optimal payment under noisy signals, \citet{mussa1978monopoly} treats contract design when creators differ, and \citet{guruganesh2021contracts} extends both algorithmically. Most recently, \citet{zhang2025fairshare} propose data pricing for general LLM data market, assuming risk-neutral creators and a reliable attribution score. None of these works combine ML data-attribution methods with contract design, nor do they treat signal informativeness as a heterogeneous quantity across creators that the payment rule must explicitly accommodate.

\textbf{Data attribution and music attribution.}
General data-attribution methods estimate per-datum contribution to model outputs. Data Shapley \citep{ghorbani2019data,jia2019towards} provides an axiomatic per-datum attribution; influence functions \citep{koh2017understanding,giordano2019swiss} and TracIn \citep{pruthi2020estimating} trace predictions to training examples; datamodels \citep{ilyas2022datamodels} and TRAK \citep{park2023trak} fit linear surrogates for the data-to-output mapping. Recent extensions scale these methods to approximate unrolled differentiation \citep{bae2024training}, datamodel-based selection \citep{engstrom2024dsdm}, LLM-scale influence \citep{grosse2023studying}, LoRA-tuned models \citep{kwon2024datainf}, and diffusion models \citep{georgiev2023journeytrak}. In the music domain specifically, fingerprinting \citep{wang2003shazam} catches near-exact matches, embedding methods \citep{spijkervet2021clmr,serra2010audio} measure distributional similarity, and \citet{barnett2024embeddings,choi2025unlearning,deng2023copyright} apply attribution to generative music for royalty assignment. This past work broadly seeks to improve informativeness in data attribution; we instead factor its (often imperfect) informativeness into a payment structure for music data.

\textbf{Generative music models.}
A growing line of generative music models are trained on large catalogs of audio recordings. Modern approaches involve two families of methods: codec LMs \citep{vandenoord2017vqvae} are language models trained on tokenized music representations \citep{dieleman2018challenge,dhariwal2020jukebox,agostinelli2023musiclm,copet2023simple}, while other approaches train diffusion models on latent audio representations \citep{ho2020ddpm,rombach2022highresolution,forsgren2022riffusion,evans2024stable}.
Output quality from these methods is now sufficient for substantial commercial value and competition with human artists, e.g., music AI startup Suno recently reported $300$m in annual recurring revenue \citep{silberling2026suno}, and a streaming service Deezer reported that $44\%$ of newly uploaded tracks are AI generated \citep{deezer2026aiuploads}.
Recent open-weight models \citep{caillon2025livemusic,gong2026acestep} have broadened access to high quality music generation outside major labs.
The increasing value and proliferation of this technology motivates urgent study of mechanisms for economic and creative sustainability, such as the payment mechanisms that we consider here.
% CHRIS: don't want to sound too biased here. this sort of implies infringement, while several of these works license their training catalogs
%copyrighted

%----------------------------------------------------------------------

%% file: section_arxiv/03_theory_v2.tex
% !TEX root = ../main.tex
\section{A framework for attribution-aware compensation}
\label{sec:theory}
%----------------------------------------------------------------------

Here we formalize our proposed music AI data payment framework as a market solution that incorporates creators and AI platform.
Beyond characterizing the equilibrium payment rule, the framework reveals how attribution shapes contracts, welfare, and competition across the market.
%\cd{I think we need to emphasize here that this is a market solution that considers creators and platforms on a neutral basis, not a regulatory one}\lz{done.}

% \nolinenumbers around the figure prevents a lineno + hyperref + figure-float
% interaction that triggers an `endlink`/`startlink` nesting bug in pdftex.
\nolinenumbers
\begin{figure*}[t]
\centering
\includegraphics[width=\textwidth]{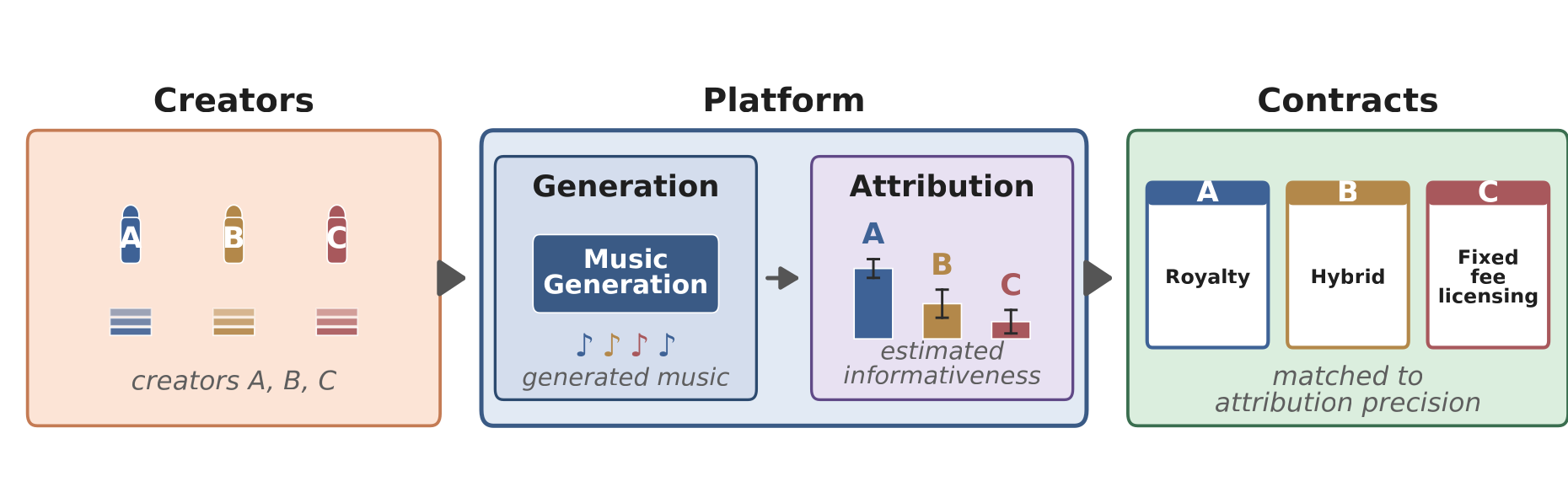}
\vspace{-25pt}
\caption{The attribution-aware compensation framework for generative music. Each creator's contract form follows from how precisely their contribution can be attributed to a generated output.}
\label{fig:framework}
\vspace{-15pt}
\end{figure*}
\linenumbers

\subsection{Framework for musical market and data attribution}
\label{sec:setup}

We model a single music AI platform that licenses training catalogs from many creators and pays each through a contract whose form depends on how well the platform can attribute generated outputs back to that creator's catalog.
The setup has four primitives --- the market and its creators, the data-attribution layers, an informativeness measure, and the contract space with creator preferences over fixed-fee and royalty bundles.

\textbf{Market and Creators.}
The market has three participants --- $N$ creators, an AI platform, and its users.
We focus on the platform--creator interaction, taking user demand as exogenous (observable from platform logs).
The platform operates a generative music service.
A user request $r$ arrives, and the platform earns gross utility $u(x,r)$ where $x\in\{0,1\}^N$ indicates the catalogs active in the training set.
User requests $r$ are drawn from a distribution $\mathcal{D}$ that the platform can estimate empirically from request logs; we make no iid assumption across requests.

Each creator $j\in\{1,\ldots,N\}$ controls a distinct, non-overlapping catalog $C_j$ (e.g.,~an artist's discography, or a label's collection) to be licensed wholesale --- all or nothing.
%We treat catalogs atomically because each creator has a distinct, non-overlapping catalog. 
%\cd{this doesn't sanity check. as a user, I can play a single track on Spotify which triggers a direct royalty. let's focus instead on the payee justification I added to the intro? that each payee has a distinct, non-overlapping catalog}\lz{done.}
We further treat attribution informativeness as accessible to both creators and the platform.
%\cd{Bookmark to revisit this... I think we should explain here that we later relax this assumption? otherwise I'd be very skeptical at this point as a reader}\lz{done --- consolidated the two-regimes discussion into the later ``We analyze two cases regarding information transparency'' paragraph (a few paragraphs below in the same Setup subsection), where it naturally belongs; removed the early forward pointer here per CD's follow-up that the duplication felt redundant. A skeptical reader encounters the full symmetric/asymmetric framing within reading distance.}
The underlying methods (EKFAC, D-TRAK, TRAK) are publicly available open-source academic tools, so creators with model access (or third-party auditors acting on their behalf, e.g., royalty collection societies or music-tech consortia) can independently verify the platform's measurements rather than rely on its claims.

Each creator has a risk-aversion parameter $\alpha_j>0$ --- how strongly they prefer steady payments to volatile ones \citep{pratt1964risk}. The platform's utility $u$ is non-additive across creators, reflecting the non-additive nature of data attribution (two creators with similar styles add less joint value than two distinct ones).

\textbf{Data attribution.}
The central informational primitive is attribution --- how much each creator's catalog contributed to a generated output.
The latent true attribution $a^*_{j}\in\R$ is the counterfactual contribution of catalog $C_j$ to the platform's outputs over a billing period (e.g., monthly in streaming).
The measured attribution $\hat a_{j}$ is a noisy estimate produced by an ML method (e.g., data Shapley \citep{ghorbani2019data}, influence functions \citep{koh2017understanding}, audio fingerprinting), and is what the platform pays on.
Unlike most data-pricing settings \citep{agarwal2019marketplace,chen2022selling,zhang2025fairshare}, music has a perceptual property --- listeners can recognize an artist's style in generated output, so $a^*_{j}$ 
%is a meaningful quantity to attribute, 
may be semantically meaningful, 
not a marginal predictive effect.
% We analyze two cases --- \emph{information symmetry}, where the platform and creators share the same attribution score (the baseline used in \Cref{sec:optimal}), and \emph{information asymmetry}, where the platform and creators do not have equal access to the score (addressed in \Cref{sec:tier3}). 

We analyze two cases regarding information transparency. \emph{Information symmetry} is the baseline in information economics \citep{akerlof1970lemons,holmstrom1979moral,hirshleifer1971private,bergemann2018design}, aligned with music-industry practice as audio fingerprinting \citep{wang2003shazam}, open-weight models, and third-party auditors (collection societies, music-tech consortia) enable independent verification of attribution scores.
\emph{Information asymmetry} is a more practical case, where platforms use proprietary or white-box attribution methods (like influence functions or EKFAC) that creators cannot reproduce without model access. The welfare gap between the two cases (\Cref{sec:tier3}) measures how much attribution transparency is worth to creators and the platform.
%\cd{above, we're strongly emphasizing the information symmetry scenario. here, this seems a bit redundant in the best case, or contradictory in the worst case?}\lz{done --- the early ``We analyze two information regimes'' paragraph that was over-emphasizing symmetric was removed; this paragraph is now the single canonical location for the two-cases framing, so there is no longer redundancy or apparent contradiction with anything earlier.}

%\cd{can we explain why these might both be relevant in the real world? if we're talking about white-box attribution, it seems like only the information asymmetry scenario is relevant. I guess if we're talking about fingerprinting, the symmetry scenario is relevant? I'd include concrete examples of each}\lz{done.}

\textbf{Informativeness.}
The quantity used by the attribution-aligned royalty benchmark (\Cref{thm:optimal_rho}) is the \emph{informativeness} $\mathcal{I}_j$ of the payment signal --- how well the noisy measurement tracks the catalog's true value to the platform.
In plain terms, $\mathcal{I}_j$ is the OLS slope of true contribution on measured attribution. Under classical measurement error ($\hat a_j=a^*_j+\xi_j$, $\xi_j\perp a^*_j$, $\E[\xi_j]=0$), $\mathcal{I}_j\in[0,1]$ --- taking the value $0$ when the signal carries no information about value and $1$ when it is a perfect linear predictor. Multiplicative scale bias or sign-flipped calibration can yield $\widehat{\mathcal{I}}_j\notin[0,1]$; we treat such cases as calibration flags rather than as violations of the definition.
%\cd{I don't understand why this is a function of payments / revenue $\pi$, rather than just $a^*_{j},\hat a_{j}$. can we factor out $\pi$? would this disrupt the rest of the framework?}\lz{done --- refactored definition to pure attribution form $\Cov(\hat a_j, a^*_j)/\Var(\hat a_j)$, with $\pi\perp(\hat a_j, a^*_j)$ now an explicit model assumption (was implicit in proofs); Theorem 1 and Prop 1 unchanged in statement, proofs updated to invoke the assumption explicitly.}
Formally,
\begin{equation}
\label{eq:informativeness}
\mathcal{I}_j:=\frac{\Cov(\hat a_{j},\;a^*_{j})}{\Var(\hat a_{j})},
\end{equation}
which depends only on the joint distribution of attribution scores $(\hat a_j, a^*_j)$, not on platform revenue.
The definition requires only finite second moments and the regression-slope identification of $a^*_j$ on $\hat a_j$; the $[0,1]$ bound is conditional on the classical-noise regime noted above.
Period-total revenue $\pi$ is the dollar revenue base on which royalties are calculated.
%\cd{two key questions about $\pi$. first, is it revenue per \emph{output}, or revenue over all outputs in some time period? if the former, this is potentially problematic, because our contract formulation below is per-creator, not per-request. also, if $\pi$ is per-request, then fixed licensing fees $f_j$ are a bit weird as usually they would be independent of platform revenue or user requests. second, what is the relationship between $\pi$ and utility $u(x, r)$ defined above? is $\pi$ supposed to be revenue, and utility is profit?}\lz{done --- (Q1) $\pi$ is period-total. (Q2) $\pi = u(x,\cdot)$ summed over the period in the default case, or the revenue subset of $u$ when platform value also includes non-monetary outcomes. Per-period scope makes $f_j$ a one-time licensing fee.}
We consider $\pi$ as $u(x,\cdot)$ summed over the period, taking $u$ as the platform's gross per-request revenue.
The framework also covers cases where the platform values non-monetary outcomes such as user retention or market position, in which case $\pi$ is the revenue component within $u(x,\cdot)$.
We adopt the centering convention $\E[\hat a_j]=\E[a^*_j]=0$ \citep{holmstrom1979moral} and assume $\pi\perp(\hat a_j, a^*_j)$, letting us factor $\pi$ out of payment-variance calculations.
% CHRIS: TODO for post deadline, can we run a normality test on $\xi_j$ in our empirical measurements to examine this assumption?
Under additive Gaussian measurement noise $\hat a_{j}=a^*_{j}+\xi_{j}$ with $\xi_{j}\sim\mathcal{N}(0,\sigma^2_{\xi,j})$ independent of $a^*_{j}$, $\mathcal{I}_j$ reduces to the signal-to-noise-ratio (SNR) form $\mathcal{I}_j=\SNR_j/(1+\SNR_j)$, where $\SNR_j:=\Var(a^*_{j})/\sigma^2_{\xi,j}$ \citep{scharf1991statistical}.
The SNR is a foundational measurement-quality metric, widely used across signal processing and information theory \citep{shannon1948mathematical,cover2006elements}.
We use this Gaussian parametrization for examples and simulations; the structural results (\Cref{thm:optimal_rho} and \Cref{prop:welfare_loss}) hold for any noise distribution with finite second moments via the covariance form in eq.~\eqref{eq:informativeness}.

\textbf{Contracts and preferences.}
The platform--creator interaction is described by three primitives --- the contract space, the creator's utility, and outside option.
The platform offers each creator a linear contract $\theta_j=(f_j,\rho_j)$ --- a fixed fee $f_j\geq 0$ 
%CHRIS: note for the future: maybe we should relax assumption and say $f_j \in \mathbb{R}$? it's concevable to me that platforms could actually *charge* creators to include their data in royalties are sufficiently rewarding
and a royalty rate $\rho_j\in[0,1]$ --- with total payment
\begin{equation}
\label{eq:payment}
w_{j}(\theta_j)=x_{j}\,[f_j+\rho_j\,v_{j}],\qquad v_{j}=\hat a_{j}\,\pi,
\end{equation}
%\cd{I wonder if the following would be more intuitive? I guess the idea is that $x_j=0\implies v_j=0$, so it's equivalent? $w_j(\theta_j)=x_j[f_j+\rho_j v_j]$, $v_j=\hat a_j\pi$.}\lz{done --- adopted the bracketed form above. Equivalent under the implicit assumption $\hat a_j=0$ when $x_j=0$, but more robust (no implicit assumption needed) and clearer; all downstream results unaffected (Thm~1, Prop~1, IC, competition, bargaining, dynamics all condition on the active case $x_j=1$).}
% Note: keep CD's alternative as inline math above; do NOT wrap a \begin{equation}...\end{equation} block inside \cd{...} --- that combination triggers a pdftex infinite-loop on TL2025.
nesting three canonical industry forms (pure fixed fee with $\rho_j=0$, pure royalty with $f_j=0$, and hybrid with $f_j>0,\rho_j>0$).
Creator $j$ evaluates contracts using mean--variance utility --- the creator values expected payment minus a multiple of payment variance \citep{holmstrom1979moral},
\begin{equation}
\label{eq:creator_utility}
U_j(\theta_j)=f_j+\rho_j\,\E[v_{j}]-\frac{\alpha_j}{2}\rho_j^2\,\Var(v_{j}).
\end{equation}
This preference primitive matches a music-industry fact --- liquidity-constrained creators value payment-stream stability, not just expected payment \citep{markowitz1952portfolio}.
Creator $j$ participates if $U_j(\theta_j)\geq\bar U_j$, where $\bar U_j$ is the utility the creator would obtain by not licensing to the platform (e.g., distributing independently).
As more catalogs are licensed and model quality improves, AI-generated music substitutes for the creator's standalone work.
The outside option therefore shrinks, ${\bar U_j=\bar U_j^0-\delta_j(x_{-j},q)}$, where $\bar U_j^0$ is the creator's standalone revenue and $\delta_j$ captures the displacement loss.

\subsection{Equilibrium contract and welfare cost}
\label{sec:optimal}

Now we separate the static profit-max contract for a single creator from an attribution-aligned benchmark for royalty design.
%\cd{why platform-optimal? do we want to mention we are assuming a regulatory climate that favors platforms? would the result differ if we examine creator-optimal?}\lz{done --- replaced 'platform-optimal' with 'equilibrium' to reflect regime-neutrality, matching the neutral-footing framing in the section opener; the royalty rate is invariant across platform-optimal, creator-optimal, and Nash-bargaining setups (only the fixed fee shifts; see \Cref{app:bargaining}).}
Both expressions reduce to closed forms in informativeness $\mathcal{I}_j$ and risk aversion $\alpha_j$, which lets us read off when the benchmark recommends royalty-dominant versus fixed-fee-dominant payment and how much alignment-relevant surplus is exposed to imperfect attribution.

We first state the platform's static licensing problem for a single creator.
The platform maximizes
\begin{equation}
\label{eq:platform_obj}
\max_{\theta_j=(f_j,\rho_j)}\;\E[\pi\,a^*_{j}]-f_j-\rho_j\,\E[v_{j}]\quad\text{s.t.}\quad U_j(\theta_j)\geq\bar U_j.
\end{equation}
The tension is the \emph{informational wedge} --- the platform pays based on the noisy attributed revenue $v_{j}=\hat a_{j}\pi$ but receives value from $a^*_{j}\pi$. Under~\eqref{eq:platform_obj} alone with binding IR, the royalty-induced risk premium is shifted back onto the platform through the fixed fee, leaving no first-order profit benefit to a positive $\rho_j$ (\Cref{lem:ir_only}). We therefore characterize an \emph{attribution-aligned royalty benchmark}: the royalty rate that balances payment--contribution co-movement against the risk premium induced by noisy royalty exposure. This benchmark is relevant when the platform, an auditor, or a policy environment values payment--contribution alignment in addition to expected net surplus.
At the optimum the participation constraint binds, so $f_j=\bar U_j-\rho_j\E[v_{j}]+(\alpha_j/2)\rho_j^2\Var(v_{j})$ and the problem reduces to a one-dimensional optimization in $\rho_j$.

\begin{lemma}[Static profit-max baseline]
\label{lem:ir_only}
Under~\eqref{eq:platform_obj} with binding IR, substituting the binding-IR fee $f_j=\bar U_j-\rho_j\E[v_j]+(\alpha_j/2)\rho_j^2\Var(v_j)$ yields
\[
\Pi_j(\rho_j)=\E[\pi\,a^*_j]-\bar U_j-\tfrac{\alpha_j}{2}\rho_j^2\Var(v_j),
\]
strictly decreasing in $\rho_j^2$ whenever $\Var(v_j)>0$. The unique maximizer is the pure fixed-fee contract $(\rho_j^*,f_j^*)=(0,\bar U_j)$. A non-degenerate royalty therefore requires an alignment criterion beyond expected net surplus; \Cref{thm:optimal_rho} characterizes the benchmark rate under that criterion.
\end{lemma}

\begin{theorem}[Attribution-aligned royalty benchmark]
\label{thm:optimal_rho}
The \emph{risk-adjusted alignment criterion}
\begin{equation}
\label{eq:alignment_criterion}
\max_{\rho_j\in[0,1]}\;\rho_j\,\Cov(v_j,a^*_j\pi)\;-\;\tfrac{\alpha_j}{2}\rho_j^2\Var(v_j),
\end{equation}
which treats payment--contribution covariance as a design target and penalizes the risk premium induced by noisy royalty exposure, admits the unique maximizer
\begin{equation}
\label{eq:optimal_rho}
\rho_j^*=\operatorname{clip}\!\Big(\frac{\mathcal{I}_j}{\alpha_j},\;[0,1]\Big).
\end{equation}
The interior solution reduces to $\rho_j^*=\mathcal{I}_j/\alpha_j$; under additive Gaussian noise it further reduces to $\rho_j^*=\SNR_j/[\alpha_j(1+\SNR_j)]$. The lower boundary binds when $\mathcal{I}_j/\alpha_j<0$, and the upper boundary binds when $\mathcal{I}_j/\alpha_j>1$. Empirical sign flips (negative $\widehat K_j$, equivalently $\widehat{\mathcal{I}}_j<0$; observed for 6/22 creators under EKFAC and 3/22 under D-TRAK; see Table~\ref{tab:headline}) and $\widehat{\mathcal{I}}_j>1$ (over-amplified signal) are reported as calibration flags. The criterion~\eqref{eq:alignment_criterion} should be read as an explicit benchmark, not as the maximizer of~\eqref{eq:platform_obj}; it is motivated by Holmstrom's informativeness principle, while \Cref{app:quality} discusses a complementary endogenous-quality channel. Proof in \Cref{app:proofs}.
\end{theorem}
The benchmark rate balances alignment against risk transfer.
$\rho_j^*\to 1/\alpha_j$ as $\mathcal{I}_j\to 1$ (royalty-dominant when attribution is precise), and $\rho_j^*\to 0$ as $\mathcal{I}_j\to 0$ (fixed-fee-dominant when attribution is too noisy).
If the benchmark is implemented through the linear contract, the fixed fee is pinned by the binding participation constraint, $f_j^*=\bar U_j-\rho_j^*\E[v_{j}]+(\alpha_j/2)(\rho_j^*)^2\Var(v_{j})$.
Attribution precision therefore determines the benchmark split between fixed fee and royalty, while the outside option determines the transfer level needed for participation.
Because the fixed fee implements the transfer level, surplus-splitting variants that share the same alignment criterion leave the benchmark royalty $\rho_j^*$ unchanged; asymmetric-criterion extensions are deferred to a future revision.

\Cref{thm:optimal_rho} yields royalty-dominant benchmark contracts for high-informativeness creators, hybrid benchmark contracts at mid informativeness, and fixed-fee-dominant benchmark contracts at low informativeness; \Cref{sec:empirical} calibrates these regimes on measured creator pools.
The wedge between the noisy signal and the true value also creates an alignment loss: noisy attribution forces any royalty component to load risk onto a signal that imperfectly tracks the creator's true contribution.

\begin{proposition}[Loss under the attribution-aligned benchmark]
\label{prop:welfare_loss}
Under the attribution-aligned benchmark, the loss relative to the perfect-attribution benchmark ($\mathcal{I}_j=1$) is
\begin{equation}
\label{eq:welfare_loss}
\mathcal{L}_j=\frac{\Var(v_{j})}{2\alpha_j}\,(1-\mathcal{I}_j)^2,
\end{equation}
and the marginal welfare gain from improving informativeness is
\begin{equation}
\label{eq:welfare_loss_marginal_general}
\left|\frac{\partial\mathcal{L}_j}{\partial\mathcal{I}_j}\right|=\frac{\Var(v_{j})}{\alpha_j}\,(1-\mathcal{I}_j),
\end{equation}
which is largest for the noisiest creators (those with $\mathcal{I}_j$ farthest from $1$). Gaussian-shortcut forms with $\SNR_j$ appear in \Cref{app:assumptions}. Proof in \Cref{app:proofs}.
\end{proposition}
This benchmark loss is concentrated among low-informativeness creators (the convex $(1-\mathcal{I}_j)^2$ scaling), and the water-filling allocation in \Cref{app:investment} routes the bulk of a marginal attribution-improvement budget to them.
When the benchmark is implemented through the linear contract, the creator is compensated for royalty risk through the fixed fee, so the loss appears as surplus that the platform must spend on risk compensation rather than as value either side can recover from a noisy signal.
Because $\mathcal{L}_j$ is computable from measurable quantities, it provides an attribution-method comparison criterion grounded in cardinal calibration rather than rank correlation alone (\Cref{sec:empirical}).

\textbf{Relation to classical contract theory.}
The signal-informativeness logic behind \Cref{thm:optimal_rho} is classical \citep{holmstrom1979moral, mussa1978monopoly, maskin1984monopoly}; what is new is that the noisy signal is itself an ML artifact the platform can invest in and creators can post-process.
Linearity is justified under mean-variance utility with Gaussian noise \citep{holmstrom1987aggregation}; alternatives in \Cref{app:discussion}.

% =====================================================================
% NEW DRAFT 3.3 --- stakeholder-organized rewrite. Theorem/proposition
% labels are reused from the existing 3.3 below (will cause "multiply
% defined label" warnings until the old subsection is removed manually).
% Apply manual inspection + cleanup after comparing the two versions.
% =====================================================================
\subsection{Market and policy implications}
\label{sec:implications}

%\cd{this section is super interesting, but also a bit impenetrable. how can we make this easier to process for particular readers? to this end, both here and in the intro, I wonder if we can reorganize these findings in a stakeholder-centric manner. e.g., For platforms, our framework suggests \ldots (e.g., attribution moat stuff). For AI researchers, \ldots (e.g., evaluating correlation not just rank correlation). For creators, our empirical results suggest \ldots (e.g., breaking from tradition and advocating for per-creator contracts). For regulatory agencies or pro organizations, our framework suggests \ldots}
%\lz{done --- reorganized into four stakeholder blocks (Platforms, Creators, AI researchers, Regulation), with three supporting results (Theorem~2 on Bertrand equilibrium, Theorem~3 on effort threshold, Proposition~5 on belief distortion) and the multi-platform setup paragraph moved to the appendix. Three core results (\Cref{prop:moat}, \Cref{prop:selective_exit}, \Cref{prop:tier3}) stay inline. Em-dashes replaced throughout the prose; forward-refs to main text removed.}

In real markets, multiple platforms compete for the same catalogs, the same pricing scheme often applies to many creators, creators choose effort levels, and platforms and creators may not see the same attribution quality. This subsection translates the framework into stakeholder-specific insights, covering what platforms should invest in, what creators should expect, what AI researchers should target, and what regulators should build into institutional infrastructure.

\textbf{Platforms.}
\label{sec:moat}
Attribution quality is a feature competitors cannot match by lowering fees. The platform with the most precise attribution holds a sustained advantage; a trailing platform earns nothing on its attribution investment until it overtakes the leader, even though the market as a whole gains throughout. Investing in better attribution is therefore the main way to build a lasting lead; lowering fees only changes who captures the value that better attribution already created.

\begin{proposition}[Attribution moat]
\label{prop:moat}
Consider $K\geq 2$ platforms competing for catalogs under exclusive licensing, in the equilibrium where each platform's royalty matches its own informativeness (formalized as \Cref{thm:bertrand} in \Cref{app:competition}). Assume Gaussian noise and homogeneous demand ($\mu_j^{(k)}=\mu_j$). Let $\Delta\mathcal{I}_j=\mathcal{I}^{(1)}-\mathcal{I}^{(2)}>0$ denote the gap between leader and trailer, and $\sigma_{a^*}^2:=\Var(a^*_{j})$. Then (i) for $\epsilon\in[0,\Delta\mathcal{I}_j]$, Platform~2's private return to its own attribution improvement is exactly zero (\emph{dead zone}); (ii) for $\epsilon>\Delta\mathcal{I}_j$, Platform~2 captures the creator and earns $\Pi_j^{(2)}(\epsilon)=(\epsilon-\Delta\mathcal{I}_j)\sigma_{a^*}^2\E[\pi^2]/(2\alpha_j)$; (iii) total welfare is weakly increasing in $\epsilon$ throughout, so the social return is strictly positive while the private return is zero up to the gap, revealing \emph{systematic underinvestment} in attribution by trailing platforms. Proof and welfare bounds in \Cref{app:competition}.
\end{proposition}
The reason why this effect is especially severe in music because catalogs are typically licensed wholesale. A trailing platform 
% CHRIS: softened the language here since this is more of an assumption than a fact
%cannot 
may be unable to 
bid for a fraction of a catalog to capture partial value; it earns zero until its attribution precision overtakes the leader's.

\textbf{Creators.}
Two compensation patterns systematically reduce creator welfare.
%\cd{wait, what are the two? having trouble parsing the second one from the following language...}\lz{done --- the two patterns are now explicitly marked First/Second; pattern 1 (uniform contracts) and pattern 2 (excessive royalty rates) are easy to identify rather than blended into the prose.}
First, the industry default of uniform compensation policies (Spotify's per-stream pro-rata, flat-fee blanket licenses, regulator-set statutory rates) disproportionately hurts creators with distinctive styles.
These are exactly the ones most costly for platforms to lose.
A distinctive creator's contribution cannot be reproduced by blending other creators' work, so their exit removes value the platform cannot recover by reweighting what remains.
Losses compound when several distinctive creators exit, costing more than the sum of their individual exits would suggest.

Second, even within personalized contracts, a sufficiently high royalty rate can reduce creator effort, because the volatility risk eventually outweighs the higher expected reward (formalized in \Cref{app:quality}).
The contracts of industry norms (``one rate for everyone'') and what looks generous on paper (``high per-stream pay'') can be the ones that reduce welfare the most.

\begin{proposition}[Selective exit under uniform contracts]
\label{prop:selective_exit}
Suppose the platform offers a uniform contract $(\bar f,\bar\rho)$ to every creator and let $\mathcal{S}^{\mathrm{exit}}=\{j:U_j(\bar f,\bar\rho)<\bar U_j\}$ be the exit set. If $\bar\rho$ is set to the population average of the creator-specific optima, then (i) creators with the highest outside options $\bar U_j$ exit first, and because distinctive creators tend to have both high $\bar U_j$ and high complementary value in $u$, $\mathcal{S}^{\mathrm{exit}}$ is disproportionately composed of the creators whose loss is most costly to the platform; (ii) when $u$ is non-additive, the welfare loss from the exits is superlinear in $|\mathcal{S}^{\mathrm{exit}}|$,
\begin{equation}
\label{eq:superlinear}
u(x,r)-u(x\setminus\mathcal{S}^{\mathrm{exit}},r)>\sum_{j\in\mathcal{S}^{\mathrm{exit}}}\big[u(x,r)-u(x\setminus\{j\},r)\big].
\end{equation}
Proof, the total-welfare-gap decomposition, and comparative statics on $\bar\rho$ in \Cref{app:uniform}.
\end{proposition}
Each additional exit costs more than the previous because distinctive musical styles cannot be replaced. In fields where contributors are interchangeable, losing one costs only what that single contributor was worth, but in music a high-informativeness creator's distinctive style cannot be rebuilt from others, so their exit eliminates every musical combination their style enabled.

\textbf{AI researchers.}
The framework gives attribution-method designers two targets. First, two methods that rank creators identically can prescribe payments differing by orders of magnitude, so numerical scores matter too, not just rankings. The welfare loss formula $\mathcal{L}_j$ (\Cref{prop:welfare_loss}) measures the gap. Methods used to set creator payments should therefore be evaluated on welfare, not just rank correlation.
Second, some attribution methods let creators cheat in a specific way. A creator can report scores that look more uniform across outputs than the true scores actually are; this benefits the creator because the platform pays a premium for steady, low-risk contributions. No filter or detection step added to the contract can stop this manipulation. The fix has to come from the attribution method itself, by choosing methods whose mathematical structure prevents this kind of smoothing.% in the first place.

\textbf{Regulation and institutional design.}
\label{sec:belief_asymmetry_implication}
\label{sec:tier3}
There are two market failures contracts cannot fix. When creators and platforms hold different beliefs about attribution quality, creator effort goes off-target and the resulting welfare loss cannot be undone by adjusting the contract. Only transparency about the platform's measured $\widehat{\mathcal{I}}_j$ closes the belief gap directly. Without transparency, creators who overestimate their informativeness over-supply effort while the platform profits from their oversupply; creators who underestimate exit instead, costing the platform the surplus their participation would have created. Manipulating attribution scores by shrinking the spread of reports (without changing the average) is impossible to deter under any linear contract; the only fix is institutional. The attribution pipeline must 
%be run 
have oversight 
by someone other than the creator (third-party auditor, collection society, or regulator-mandated infrastructure). Together these results mark where contract design hits its limit and institutional intervention becomes necessary, not optional.

\begin{proposition}[Score-smoothing manipulation cannot be deterred]
\label{prop:tier3}
Under $\hat a_j\perp\pi$ (or $\pi\equiv 1$), for any mean-preserving shrinkage (compressing the spread of reported scores toward their mean while keeping the average unchanged) $g_\epsilon(\hat a)=(1-\epsilon)\hat a+\epsilon\bar a$ with $\epsilon\in(0,1]$, the creator's utility change is
\begin{equation}
\label{eq:tier3_gain}
\Delta U_j^{\mathrm{shrink}}(\epsilon)=\frac{\alpha_j}{2}(\rho_j^*)^2\,\epsilon(2-\epsilon)\,\Var(\hat a_j)\,\E[\pi^2]\;>\;0.
\end{equation}
Profitable for every $\rho_j>0$ and therefore \emph{impossible to deter under any linear contract}; full IC analysis in \Cref{app:ic}.
\end{proposition}
Royalty contracts work by tying payment to the attribution score. That same dependency becomes a vulnerability when creators can change the score before reporting it. The fix is for the platform to compute attribution itself, rather than trust the numbers creators report.

The framework converts the cost of inaccurate attribution into a concrete per-creator number $\mathcal{L}_j$, which can be measured, compared across attribution methods, and used to direct investment where it helps most (\Cref{cor:waterfilling} in \Cref{app:investment}). The largest welfare gains from improving attribution go to creators who look least valuable under today's noisy measurement.

%----------------------------------------------------------------------

%% file: section_arxiv/05_empirical.tex
% !TEX root = ../main.tex
\section{Empirical instrument and measurement}
\label{sec:empirical}
%----------------------------------------------------------------------

Data-attribution methods are usually evaluated by rank correlation, but the welfare formula in \S3 uses cardinal calibration, not just ranks. We apply the framework's instrument across two settings, audio latent diffusion (Stable Audio Open~\citep{evans2024stable}) with EKFAC~\citep{mlodozeniec2025influence} and D-TRAK~\citep{zheng2024dtrak}, and a symbolic LM (Anticipatory Music Transformer~\citep{thickstun2023anticipatory}) with LoGra~\citep{choe2024your} and TRAK~\citep{park2023trak}, and ask whether current attribution delivers what royalty regimes need. A market-simulation companion in \Cref{app:market_simulation} draws priors from the same measurement.

\subsection{Setting 1 --- Latent-diffusion audio attribution}
\label{sec:emp_attribution}

\textbf{Setup.}
We fine-tune Stable Audio Open \citep{evans2024stable} on $22$ artist catalogs from MTG-Jamendo \citep{bogdanov2019mtg}, retraining $22$ leave-catalog-out variants alongside one full-catalog model, and evaluate on $200$ generated outputs.
Ground truth is the loss-delta on each generation when its creator's catalog is held out of training (formal definition and sanity checks in \Cref{app:setup_details}).
We benchmark EKFAC \citep{mlodozeniec2025influence} and D-TRAK \citep{zheng2024dtrak} and report three quantities per (method, creator).
These are informativeness $\widehat{\mathcal{I}}_j$ from \S3, calibration scale $\widehat K_j$ (regression slope of measured attribution on truth), and calibrated SNR $\widehat{\SNR}^{\mathrm{cal}}_j=\widehat{\Var}(a^*_j)/\widehat{\Var}(\hat a_j/\widehat K_j-a^*_j)$.

\textbf{Cardinal measurement uncovers signal rank metrics miss.}
Both audio methods produce a small but real cardinal signal against the ground truth, several-fold stronger than rank correlations alone (\Cref{fig:results_main}(a); per-method numbers in \Cref{tab:headline}).
Rank-based evaluation \citep{park2023trak,ilyas2022datamodels} discards exactly the cardinal calibration that the welfare loss formula $\mathcal{L}_j$ in \Cref{prop:welfare_loss} uses.

\textbf{Rank-equivalent methods disagree on calibration.}
EKFAC and D-TRAK match on Spearman but produce calibration scales $\widehat K_j$ differing by roughly $10^3\times$ (\Cref{fig:results_main}(b)).
Several creators show attribution scores with opposite signs across methods, including the pool's most genre-coherent catalog (a $31$-track classical artist) under EKFAC.
Because welfare loss $\mathcal{L}_j$ in \Cref{prop:welfare_loss} depends on calibration, methods that look equivalent under rank correlation prescribe materially different contracts.

\textbf{Current attribution closes about $10\%$ of the welfare gap.}
Current methods close only the first $10\%$ of the welfare gap that royalty-grade attribution would close.
At most $1$ of $22$ creators clears the \emph{royalty threshold} --- the precision at which royalty payment delivers more welfare than a fixed fee --- under either audio method (\Cref{fig:results_main}(c)); the other $21$ sit where each unit of precision improvement pays off most (\Cref{fig:results_main}(d)).
Per-creator headline statistics and bootstrap CIs are in \Cref{app:per_creator}.

\begin{figure*}[t]
\centering
\includegraphics[width=\textwidth]{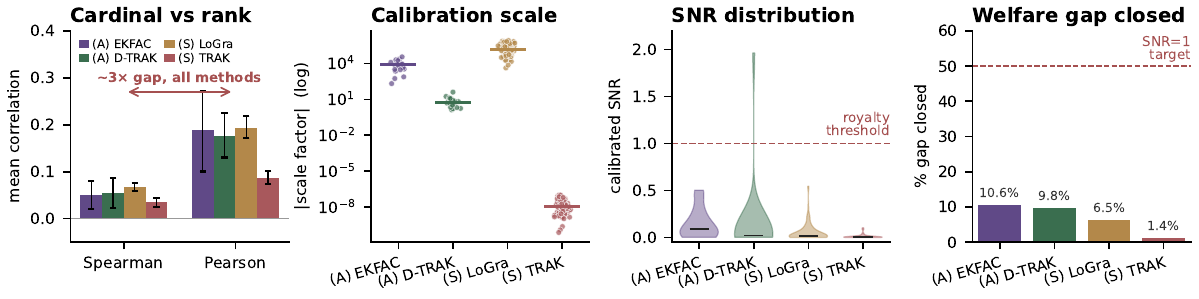}
\vspace{-15pt}
\caption{Attribution measurement on two settings, audio diffusion (Stable Audio Open with EKFAC and D-TRAK, $22$ catalogs) and symbolic language modeling (Anticipatory Music Transformer with LoGra and TRAK, $100$ artists).
Across both settings and all four methods, almost no creator's attribution is precise enough for royalty payment to beat a fixed fee under the attribution-aligned benchmark, which is fixed-fee-dominant for nearly the entire measured pool.
\textbf{(a)}~Per-query rank and cardinal correlation against the leave-catalog-out ground truth.
\textbf{(b)}~Per-creator calibration scale, log axis.
\textbf{(c)}~Calibrated-SNR distribution per method, dashed line at the royalty threshold.
\textbf{(d)}~Welfare gap closed at $\alpha{=}2$.}
\label{fig:results_main}
\vspace{-15pt}
\end{figure*}

\subsection{Setting 2 --- Symbolic music attribution}
\label{sec:emp_setting2}

We also conduct pilot measurements in the symbolic domain, consistent with parameters of a commercially deployed music co-creation product~\citep{donahue2025hookpad,kim2025amuse}. 
The pilot fine-tunes an Anticipatory Music Transformer \citep{thickstun2023anticipatory} on $100$ TheoryTab artists, generates $500$ outputs per model under the same leave-catalog-out protocol, and measures the same three quantities using LoGra \citep{choe2024your} and TRAK \citep{park2023trak}.

\textbf{Cross-domain consistency confirms a shortfall across methods.}
Two completely different settings (audio diffusion vs.\ symbolic transformer, four different attribution methods, two domains) return the same diagnosis.
Symbolic Spearman and Pearson sit in the same regime as Setting~1's, calibrated SNR is roughly an order of magnitude below the royalty threshold, and no creator clears the threshold under either symbolic method.
The attribution-precision shortfall is a property of current ML methods broadly, not an artifact of any single architecture or modality.

\textbf{The instrument discriminates between methods more sharply on symbolic.}
On symbolic data, LoGra clearly outperforms TRAK; on audio, EKFAC and D-TRAK look more comparable on rank but differ by orders of magnitude in calibration scale.
The cross-method calibration gap on symbolic ($\sim 10^{13}$) is even larger than on audio ($\sim 10^{3}$), so the rank-vs-cardinal split that the framework first flagged in Setting~1 is amplified, not erased, by switching model class.
Methods that look similar under rank correlation prescribe materially different contracts on the same pool.

\textbf{Implication for the framework's prescription.}
Both settings sit deep inside the noisy regime where each unit of attribution improvement pays off most.
The cross-setting consistency strengthens the central empirical conclusion.
Under current attribution capability, the attribution-aligned benchmark is fixed-fee-dominant across the model classes we measured, and royalty-heavy regimes require precision improvements not yet delivered by any tested method.

\subsection{Musical market simulation}
\label{sec:emp_market}

We instantiate the framework's attribution-aligned benchmark on the per-creator quantities measured above, then ask how much SNR improvement would be needed to make royalty regimes viable under that benchmark.

\textbf{Contract form across the population.}
For each creator $j$ and risk-aversion level $\alpha$ the benchmark royalty rate is $\rho^*_j = \widehat{\mathcal{I}}^{\mathrm{cal}}_j/\alpha$ after clipping to $[0,1]$ (regime classification in \Cref{fig:results_partb}(a)).
At $\alpha=2$, no creator out of $22$ clears the royalty threshold under either audio method (\Cref{fig:results_partb}(b)); the only royalty-eligible case under any setting is a $31$-track classical-music catalog at $\alpha \leq 1$ on D-TRAK.
The framework's prescription is near-universal fixed-fee at typical risk aversion, matching the dominant industry pattern of blanket licensing.
The Setting~2 (symbolic) population yields a similarly fixed-fee-dominant prescription.

\textbf{Welfare counterfactual under attribution improvements.}
We ask how much SNR improvement would be needed to close more of the gap (\Cref{fig:results_partb}(c), formula in \Cref{app:market_simulation}).
Closing $50\%$ requires $25$--$40\times$ improvement in calibrated SNR, raising the median creator from today's $\approx 0.05$ to $\approx 1$.
This is a concrete precision target for royalty-based AI music compensation to become viable.

\textbf{Multi-platform competition.}
Treating EKFAC and D-TRAK as two competing platforms in the Bertrand equilibrium of \Cref{thm:bertrand} yields a per-creator winner (EKFAC for $14$ of $22$ creators, D-TRAK for $8$; \Cref{fig:results_partb}(d)).
The per-creator surplus gap between methods realizes \Cref{prop:moat}'s underinvestment prediction on a measured pool rather than a hypothetical one.
The median gap is small enough that neither platform gains enough from improving its own attribution to overtake the other, even where investing would help society overall.
Full per-creator analysis is in \Cref{app:market_simulation}.

\begin{figure*}[t]
\centering
\includegraphics[width=\textwidth]{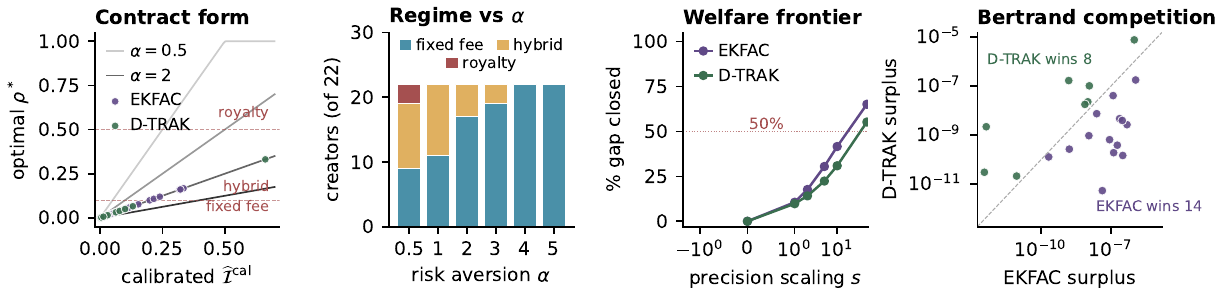}
\vspace{-15pt}
\caption{Framework prescription on the measured creator pool, $22$ catalogs under EKFAC and D-TRAK.
\textbf{(a)}~Optimal royalty rate versus calibrated informativeness at $\alpha{=}2$, with hybrid and royalty thresholds.
\textbf{(b)}~Contract-form counts across risk-aversion levels under EKFAC.
\textbf{(c)}~Welfare gap closed as calibrated SNR is scaled.
\textbf{(d)}~Per-creator surplus, EKFAC versus D-TRAK, log--log axes.}
\label{fig:results_partb}
\vspace{-15pt}
\end{figure*}

%----------------------------------------------------------------------

%% file: section_arxiv/06_conclusion.tex
% !TEX root = ../main.tex
\section{Conclusion}
\label{sec:conclusion}
%----------------------------------------------------------------------

We propose a framework for compensating music creators whose work trains generative AI, where each creator's contract form follows from how accurately the platform can trace generated outputs back to their catalog.
As generative AI takes a larger share of music production, the choice of attribution technology will increasingly shape who gets paid, and our framework gives platforms, creators, and policymakers a shared way to reason about that choice.
Future work could extend it to catalogs that grow over time, generative models spanning music and other media, and institutional setups for attribution oversight.

%% file: section_arxiv/A_gaussian.tex
% !TEX root = ../main.tex
\section{Gaussian special case and simplifying conditions}
\label{app:assumptions}

The main results (Theorem~\ref{thm:optimal_rho}, Propositions~\ref{prop:welfare_loss}--\ref{prop:selective_exit}) hold for any noise structure through the general informativeness measure $\mathcal{I}_j$. The closed-form SNR expressions used in examples and simulations require three additional conditions, stated here for reference.

\textbf{Additive Gaussian noise.} $\hat{a}_{j} = a^*_{j} + \xi_{j}$, with $\xi_{j} \sim \mathcal{N}(0, \sigma^2_{\xi,j})$ independent of $a^*_{j}$ and $\pi$. This yields $\mathcal{I}_j = \SNR_j / (1+\SNR_j)$ and the closed-form $\rho^* = \frac{1}{\alpha_j}\cdot\frac{\SNR_j}{1+\SNR_j}$.
Under the additional revenue-independence condition below, the welfare loss in eq.~\eqref{eq:welfare_loss} reduces to the Gaussian-shortcut form
\begin{equation}
\label{eq:welfare_loss_gaussian}
\mathcal{L}_j=\frac{\E[\pi^2]\,\Var(a^*_{j})}{2\alpha_j\,\SNR_j(1+\SNR_j)},
\end{equation}
and the marginal welfare gain from improving $\SNR_j$ is
\begin{equation}
\label{eq:welfare_loss_marginal}
\left|\frac{\partial\mathcal{L}_j}{\partial\SNR_j}\right|=\frac{\E[\pi^2]\,\Var(a^*_{j})\,(2\SNR_j+1)}{2\alpha_j\,\SNR_j^2(1+\SNR_j)^2},
\end{equation}
which scales as $1/\SNR_j^2$ at low precision.

\textbf{Multiplicative-bias regime and calibration.} A common generalization is $\hat{a}_{j} = K_j \cdot a^*_{j} + \xi_{j}$ with $\xi_{j} \perp a^*_{j}$ and $\xi_{j} \sim \mathcal{N}(0, \sigma^2_{\xi,j})$, which arises naturally for gradient-based attribution methods that estimate scores up to a method-specific scale (Section~\ref{sec:setup}). Under this model, the calibrated estimator $\tilde{a}_{j} = \hat{a}_{j}/K_j$ satisfies $\tilde{a}_{j} = a^*_{j} + \xi_{j}/K_j$, returning the additive Gaussian regime with $\SNR_j^{\mathrm{cal}} = K_j^2 \cdot \SNR_j$. The benchmark payment rule should be evaluated on the calibrated signal when the Gaussian shortcut is used: $\mathcal{I}^{\mathrm{cal}}_j=\SNR_j^{\mathrm{cal}}/(1+\SNR_j^{\mathrm{cal}})$, while the raw uncalibrated regression slope satisfies $\mathcal{I}^{\mathrm{raw}}_j=\SNR_j^{\mathrm{cal}}/[K_j(1+\SNR_j^{\mathrm{cal}})]$. Thus the welfare-loss formula in equation~\eqref{eq:welfare_loss} requires the calibrated $\widehat{\SNR}^{\mathrm{cal}}$ when expressed via the Gaussian shortcut. The proof is a one-line algebra exercise on $\Cov(\hat{a},a^*) = K_j \Var(a^*)$ and $\Var(\hat{a}) = K_j^2 \Var(a^*) + \Var(\xi)$.

\textbf{Revenue independence.} $a^*_{j} \perp \pi \mid x$: individual attribution does not predict total platform revenue, conditional on the active set. This simplifies $\Var(a^*_{j} \cdot \pi) = \E[\pi^2]\Var(a^*_{j})$ and is plausible when individual contributions are small relative to the full catalog.

\textbf{Non-additive utility.} The selective-exit result (Proposition~\ref{prop:selective_exit}) and the cross-creator substitution argument apply to settings where $u(x, r)$ is non-additive across creators --- some creators are complements (superadditive), others substitutes (submodular). This is a property of the platform's value function, not a distributional assumption on the attribution signal.

%% file: section_arxiv/B_proofs.tex
% !TEX root = ../main.tex
\section{Full proofs and the quality-choice extension}
\label{app:proofs}

This appendix contains the full proofs of the main-text results, followed by the formal extension to endogenous quality choice in \Cref{app:quality}, referenced from the main-text forward pointer.

\subsection{Proof of Theorem~\ref{thm:optimal_rho}}

\begin{proof}[Proof of Theorem~\ref{thm:optimal_rho}]
The theorem maximizes the risk-adjusted alignment criterion~\eqref{eq:alignment_criterion}, which is distinct from the platform-profit objective~\eqref{eq:platform_obj}. The latter is handled by \Cref{lem:ir_only} and gives the pure fixed-fee baseline.

Under the model assumptions $\pi\perp(\hat a_j,a^*_j)$ and the centering convention $\E[\hat a_j]=\E[a^*_j]=0$, independence and zero means give
\[
\Cov(v_j,a^*_j\pi)=\E[\pi^2]\Cov(\hat a_j,a^*_j),\qquad
\Var(v_j)=\E[\pi^2]\Var(\hat a_j).
\]
The first-order condition for~\eqref{eq:alignment_criterion} gives the interior maximizer
\[
\rho_j^{\mathrm{int}}
=\frac{\Cov(v_j,a^*_j\pi)}{\alpha_j\Var(v_j)}
=\frac{\Cov(\hat a_j,a^*_j)}{\alpha_j\Var(\hat a_j)}
=\frac{\mathcal{I}_j}{\alpha_j}.
\]
The second-order condition is $-\alpha_j\Var(v_j)<0$, so the criterion is strictly concave whenever $\Var(v_j)>0$. Imposing the contract-space constraint $\rho_j\in[0,1]$ projects the interior maximizer onto the feasible interval, yielding $\rho_j^*=\operatorname{clip}(\mathcal{I}_j/\alpha_j,[0,1])$.

For the Gaussian special case, $\hat a_j=a^*_j+\xi_j$ with $\xi_j$ independent of $a^*_j$ gives $\mathcal{I}_j=\SNR_j/(1+\SNR_j)$, so the interior benchmark rate is $\rho_j^*=\SNR_j/[\alpha_j(1+\SNR_j)]$.
\end{proof}

\subsection{Proof of Proposition~\ref{prop:welfare_loss}}

\begin{proof}[Proof of Proposition~\ref{prop:welfare_loss}]
For the attribution-aligned benchmark, the reduced criterion as a function of $\rho_j$ is
\[
J_j(\rho_j)=-\frac{\alpha_j}{2}\rho_j^2\Var(v_j)+\rho_j\Cov(v_j,a^*_j\pi),
\]
up to terms independent of $\rho_j$. Under the model assumptions ($\pi\perp(\hat a_j,a^*_j)$ and centering $\E[\hat a_j]=\E[a^*_j]=0$), the covariance identity $\Cov(v_j,a^*_j\pi)=\mathcal{I}_j\Var(v_j)$ holds. Completing the square around $\rho_j^*=\mathcal{I}_j/\alpha_j$ gives
\[
-\frac{\alpha_j}{2}\rho_j^2\Var(v_j)+\rho_j\mathcal{I}_j\Var(v_j)
=-\frac{\alpha_j\Var(v_j)}{2}\big(\rho_j-\rho_j^*\big)^2+\frac{\Var(v_j)}{2\alpha_j}\mathcal{I}_j^2.
\]
Thus
\begin{equation}
\label{eq:benchmark_completed_square}
J_j(\rho_j)=\frac{\Var(v_j)}{2\alpha_j}\mathcal{I}_j^2-\frac{\alpha_j\Var(v_j)}{2}\big(\rho_j-\rho_j^*\big)^2.
\end{equation}
The benchmark loss reported in the main text is the quadratic gap induced by moving from the perfect-attribution benchmark rate $\rho_j^{\mathrm{FB}}=1/\alpha_j$ to the noisy-attribution benchmark rate $\rho_j^*=\mathcal{I}_j/\alpha_j$, evaluated under the same variance normalization:
\[
\mathcal{L}_j
=\frac{\alpha_j\Var(v_j)}{2}\big(\rho_j^{\mathrm{FB}}-\rho_j^*\big)^2
=\frac{\Var(v_j)}{2\alpha_j}(1-\mathcal{I}_j)^2.
\]

In the Gaussian special case, substitute $\mathcal{I}_j=\SNR_j/(1+\SNR_j)$ and the variance convention used in the simulations to obtain the SNR form stated in \Cref{app:assumptions}.
\end{proof}

\subsection{Endogenous quality choice}
\label{app:quality}

The static framework of Section~\ref{sec:theory} takes the creator's contribution $a^*_{j}$ as exogenous. We now allow the creator to choose effort $q_j$ that scales their contribution multiplicatively, so platform value from creator $j$ is $q_j\cdot a^*_{j}\pi$.

\textbf{Setup.} Adopt the following extension of the main-text assumptions:

\textbf{(Q1) Multiplicative scaling.} Quality enters multiplicatively: platform value from creator $j$ is $q_j\cdot a^*_{j}\pi$. Under multiplicative scaling, $\mathcal{I}_j$ is invariant to $q_j$ (the $q_j^2$ cancels in $\Cov$ over $\Var$).

\textbf{(Q2) Participation.} Outside option $\bar U_j$ does not depend on $\rho_j$; the IR constraint binds at the optimum.

\textbf{(Q3) Sequential timing.} The platform commits to $\theta_j=(f_j,\rho_j)$, after which the creator chooses effort $q_j$ to maximize
\[
U_j(\theta_j,q)=f_j+\rho_j q\mu_j-\frac{A_j}{2}\rho_j^2 q^2-\frac{\gamma_j}{2}q^2,
\]
where $\mu_j:=\E[v_{j}]$, $A_j:=\alpha_j\sigma_{v,j}^2$, $D_j(\rho):=A_j\rho^2+\gamma_j$, and $\gamma_j>0$ is the marginal effort cost.

\textbf{(Q4) Unbiased attribution.} $\E[\hat a_{j}\pi]=\E[a^*_{j}\pi]$, so the platform's expected gross benefit from effort is $q_j\mu_j$.

\begin{proof}[Proof of the endogenous-quality-choice result]
The creator's objective is strictly concave in $q$ ($\partial^2 U_j/\partial q^2=-D_j<0$). The corner $q=0$ is excluded since $\partial U_j/\partial q\big|_{q=0}=\rho\mu>0$. The interior FOC gives $q^{\mathrm{int}}=\rho\mu/D$. Differentiating: $dq^*/d\rho=\mu(\gamma-A\rho^2)/D^2$, positive iff $\rho<\sqrt{\gamma/A}$. The peak is obtained by substituting $\rho^{\max}=\sqrt{\gamma/A}$ into $q^{\mathrm{int}}$. Comparative statics in $\gamma_j$ and $\alpha_j$ follow from $D$ being strictly increasing in both.
\end{proof}

\textbf{Welfare-maximizing royalty.} The hump-shape implies that the welfare-maximizing royalty---which maximizes the effort surplus $W_j(\rho)=\mu_j^2\rho(2-\rho)/[2D_j(\rho)]$---generically exceeds the alignment-only optimum $\rho_j^*=\mathcal{I}_j/\alpha_j$.

\begin{proposition}[Effort-only optimal royalty]
\label{prop:effort_rho}
The unique maximizer of $W_j$ on $(0,1)$ is
\[
\rho_j^{**}=\frac{-\gamma_j+\sqrt{\gamma_j^2+4A_j\gamma_j}}{2A_j},\qquad 0<\rho_j^{**}<1.
\]
\end{proposition}

\begin{proof}
$W_j'(\rho)=\mu_j^2 h(\rho)/D_j(\rho)^2$ where $h(\rho)=-A_j\rho^2-\gamma_j\rho+\gamma_j$. Setting $h=0$ and applying the quadratic formula yields the positive root. The intermediate-value theorem confirms the root lies in $(0,1)$ since $h(0)=\gamma_j>0$ and $h(1)=-A_j<0$.
\end{proof}

\begin{proposition}[Royalty hierarchy]
\label{prop:rho_hierarchy}
Under (Q1)--(Q4) with $\beta_j:=A_j/\gamma_j<1$ and $\rho_j^*<\rho_j^{**}$, we have $\rho_j^{**}>\rho_j^*$. A sufficient condition is $\alpha_j\geq 1$ and $\beta_j<1$. As $\gamma_j\to\infty$, $\rho_j^{**}\to 1$ and the effort channel becomes negligible: total welfare is governed by the alignment channel and the optimal payment rule converges to $\rho_j^*=\mathcal{I}_j/\alpha_j$.
\end{proposition}

\begin{proof}
Since $W_j$ is strictly increasing on $(0,\rho_j^{**})$ (from $W_j'>0$ when $h>0$) and $\rho_j^*<\rho_j^{**}$ by hypothesis, $W_j'(\rho_j^*)>0$, so raising the royalty above $\rho_j^*$ improves effort welfare. The alignment loss $\mathcal{L}_j$ is $\rho$-independent under (Q1), so there is no first-order cost. As $\gamma_j\to\infty$, $q_j^*\to 0$ and $W_j(\rho)\to 0$ uniformly; the static alignment channel dominates, and the optimum converges to $\rho_j^*$.
\end{proof}

The economic interpretation is that the alignment channel of Theorem~\ref{thm:optimal_rho} captures only the cost of measurement noise; under endogenous effort, the royalty also serves a moral-hazard role, and the welfare-optimal rate trades off the two. At pilot parameters ($\alpha=2$, $\beta\approx 0.02$), $\rho_j^{**}\approx 0.98$ for all three creators, an order of magnitude above the alignment-only $\rho_j^*\leq 0.378$.

\textbf{Note on remaining proofs.} The cross-tier substitution argument from Section~\ref{sec:setup} is in Appendix~\ref{app:uniform}, where the uniform-contract proof handles the welfare gap between personalized and uniform rates. The dynamic Bellman equation and two-regime characterization for participation appear in Appendix~\ref{app:dynamics}.

%% file: section_arxiv/C_ic.tex
% !TEX root = ../main.tex
\section{Incentive compatibility: full analysis}
\label{app:ic}

This appendix develops the full incentive-compatibility machinery underlying the score-smoothing analysis (Section~\ref{sec:tier3}). We give the proof of the three-tier hierarchy, characterize when monitoring frictions deter Tier~2 (mean inflation) bias, work out the screening problem when risk aversion $\alpha_j$ is private information, derive a complementary statistical detection test for Tier~3 (variance shrinkage), and show that the IC hierarchy is invariant to the competitive environment under exclusive licensing. Throughout we use the notation of Section~\ref{sec:setup}: $v_{j}=\hat a_{j}\pi$, $\sigma_{v,j}^2=\Var(v_{j})$, $\rho_j^*=\mathcal{I}_j/\alpha_j$, and the creator's mean--variance utility~\eqref{eq:creator_utility}. The single-platform variance $\sigma_{v,j}^2$ is the $K=1$ specialization of the per-platform variance $V_j^{(k)}=\Var(v_{j}^{(k)})$ used in Appendix~\ref{app:competition}.

\subsection{Proof of the three-tier hierarchy}
\label{app:ic_hierarchy}

We give the formal version of the three-tier hierarchy whose Tier~3 conclusion appears as Proposition~\ref{prop:tier3} in the main text.

\begin{proposition}[Three-tier creator manipulation hierarchy]
\label{prop:ic_full}
Fix $\rho_j^*=\mathcal{I}_j/\alpha_j$ from Theorem~\ref{thm:optimal_rho}. Consider three post-contractual distortions of the reported $\hat a_{j}$:

\textbf{(T1) Noise inflation.} The creator adds independent noise $\eta\perp(\hat a,a^*,\pi)$ with $\E[\eta]=0$, $\Var(\eta)=\Sigma_\eta>0$ at private cost $c\geq 0$. The creator's utility change is
\[
\Delta U_j^{\mathrm{noise}}=-\frac{\alpha_j}{2}(\rho_j^*)^2\Sigma_\eta\E[\pi^2]-c<0.
\]
Tier~1 is automatically deterred for every $\alpha_j>0$ and $\rho_j^*>0$.

\textbf{(T2) Mean inflation.} The creator adds bias $b\geq 0$ at quadratic cost $(\kappa/2)b^2$, $\kappa>0$. The interior-bias maximizer is
\[
b_j^*=\max\!\Big(0,\;\frac{\rho_j^*\bar\pi-\alpha_j(\rho_j^*)^2\Cov(v_{j},\pi)}{\kappa+\alpha_j(\rho_j^*)^2\sigma_\pi^2}\Big).
\]
Under $\pi\equiv 1$ this simplifies to $b_j^*=\mathcal{I}_j/(\alpha_j\kappa)>0$. Tier~2 is not self-policing: external monitoring ($\kappa$ large enough) is required to enforce truthfulness.

\textbf{(T3) Variance shrinkage.} Assume $\hat a_j\perp\pi$ or $\pi\equiv 1$. For any mean-preserving shrinkage $g_\epsilon(\hat a)=(1-\epsilon)\hat a+\epsilon\bar a$ with $\epsilon\in(0,1]$,
\[
\Delta U_j^{\mathrm{shrink}}(\epsilon)=\frac{\alpha_j}{2}(\rho_j^*)^2\,\epsilon(2-\epsilon)\,\Var(\hat a_j)\,\E[\pi^2]>0.
\]
Tier~3 is impossible to deter under any linear contract with $\rho_j>0$.
\end{proposition}

\begin{proof}
\textbf{(T1).} Let $\tilde v=v+\eta\pi$. By independence and $\E[\eta]=0$, $\E[\eta\pi]=\E[\eta]\E[\pi]=0$, so the mean is preserved. By the tower property, $\Cov(v,\eta\pi)=\E[\hat a\pi^2\E[\eta\mid\hat a,a^*,\pi]]=0$. Independence gives $\Var(\eta\pi)=\E[\eta^2]\E[\pi^2]=\Sigma_\eta\E[\pi^2]$. Hence $\Var(\tilde v)=\Var(v)+\Sigma_\eta\E[\pi^2]$, and the creator's utility change is $-(\alpha_j/2)(\rho_j^*)^2\Sigma_\eta\E[\pi^2]-c$, which is strictly negative.

\textbf{(T2).} Adding $b$ replaces $v$ with $v+b\pi$. The creator solves $\max_{b\geq 0}\rho_j^*\bar\pi b-(\alpha_j/2)(\rho_j^*)^2(b^2\sigma_\pi^2+2b\Cov(v,\pi))-(\kappa/2)b^2$. The objective is strictly concave (SOC: $-\kappa-\alpha_j(\rho_j^*)^2\sigma_\pi^2<0$); the interior FOC and the $b\geq 0$ constraint give the stated $b_j^*$.

\textbf{(T3).} Under $\hat a_j\perp\pi$, mean preservation: $\E[g_\epsilon(\hat a)\pi]=\bar a\bar\pi=\E[\hat a\pi]$. Apply the product-of-independents variance identity (Lemma~\ref{lem:prodvar}): $\Var(g_\epsilon(\hat a)\pi)=(1-\epsilon)^2\Var(\hat a)\E[\pi^2]+\bar a^2\Var(\pi)$ and $\Var(\hat a\pi)=\Var(\hat a)\E[\pi^2]+\bar a^2\Var(\pi)$. The variance reduction is $\epsilon(2-\epsilon)\Var(\hat a)\E[\pi^2]$; the $\bar a^2\Var(\pi)$ term cancels exactly. Mean payment is unchanged, so the utility change is $(\alpha_j/2)(\rho_j^*)^2\,\epsilon(2-\epsilon)\Var(\hat a)\E[\pi^2]>0$.
\end{proof}

\begin{lemma}[Product of independents]
\label{lem:prodvar}
For $X\perp Y$ with finite second moments, $\Var(XY)=\Var(X)\E[Y^2]+\E[X]^2\Var(Y)$.
\end{lemma}

\begin{proof}
$\Var(XY)=\E[X^2Y^2]-\E[XY]^2=\E[X^2]\E[Y^2]-\E[X]^2\E[Y]^2=(\Var(X)+\E[X]^2)\E[Y^2]-\E[X]^2\E[Y]^2=\Var(X)\E[Y^2]+\E[X]^2\Var(Y)$.
\end{proof}

\textbf{Tightness.} Each hypothesis binds. Tier~1: at $\alpha_j=0$ the variance penalty vanishes and Tier~1 fails weakly when $c=0$. Tier~2: $\kappa=0\Rightarrow b^*\to\infty$. Tier~3: $\Var(\hat a_j)=0$ removes any variance to reduce, and $\rho_j^*=0$ disconnects payment from $\hat a$. The condition $\hat a_j\perp\pi$ (or $\pi\equiv 1$) is used in Tier~3; the general-$\pi$ analysis (Remark~\ref{rmk:general_pi}) shows the SNR$^2$ ratio bound is tight in the uniform-demand limit and strictly bounded above otherwise.

\subsection{SNR-squared manipulation--welfare ratio}
\label{app:ic_snr}

\begin{corollary}[SNR$^2$ manipulation--welfare ratio]
\label{cor:snr_ratio}
Under $\pi\equiv 1$, the ratio of the maximum Tier~3 manipulation gain (at $\epsilon=1$) to the welfare loss from Proposition~\ref{prop:welfare_loss} is
\[
\frac{\Delta U_j^{\max}}{\mathcal{L}_j}=\Big(\frac{\mathcal{I}_j}{1-\mathcal{I}_j}\Big)^2=\SNR_j^2.
\]
\end{corollary}

\begin{proof}
From Proposition~\ref{prop:ic_full}, $\Delta U_j^{\max}=(\alpha_j/2)(\rho_j^*)^2\Var(\hat a_j)=\mathcal{I}_j^2\Var(\hat a_j)/(2\alpha_j)$ at $\E[\pi^2]=1$. Under $\pi\equiv 1$, $\Var(v_j)=\Var(\hat a_j)$. Substituting into Proposition~\ref{prop:welfare_loss}: $\mathcal{L}_j=\Var(\hat a_j)(1-\mathcal{I}_j)^2/(2\alpha_j)$. The ratio is $\mathcal{I}_j^2/(1-\mathcal{I}_j)^2=\SNR_j^2$.
\end{proof}

\begin{remark}[General-$\pi$ upper bound]
\label{rmk:general_pi}
When $\pi$ is stochastic and $\hat a_j\perp\pi$, Lemma~\ref{lem:prodvar} gives $\Var(v_j)=\Var(\hat a_j)\E[\pi^2]+\bar a^2\Var(\pi)\geq\Var(\hat a_j)\E[\pi^2]$, so
\[
\frac{\Delta U_j^{\max}}{\mathcal{L}_j}=\SNR_j^2\cdot\frac{\Var(\hat a_j)\E[\pi^2]}{\Var(\hat a_j)\E[\pi^2]+\bar a^2\Var(\pi)}\leq\SNR_j^2,
\]
with equality iff $\bar a=0$ or $\Var(\pi)=0$. The bound is therefore conservative in the natural direction: $\SNR_j^2$ is an upper bound on the manipulation--welfare ratio, attained in the uniform-demand limit.
\end{remark}

The policy implication is that monitoring expenditure should scale super-linearly in attribution precision: the most precisely measured creators are simultaneously the lowest-loss and the most exposed to Tier~3 manipulation.

\subsection{Statistical detection of Tier-3 manipulation}
\label{app:ic_detect}

Although Tier~3 cannot be \emph{deterred} by any linear contract, it can often be \emph{detected} from a finite history of reported $\hat a$.

\begin{remark}[Variance-ratio test]
\label{rmk:tier3_detect}
Under a Gaussian model, a variance-ratio test on the empirical distribution of reported $\hat a$ over $T$ periods has statistic
\[
T_{\mathrm{stat}}=(T-1)\widehat{\Var}(\hat a^{\mathrm{rep}})/\Var(\hat a^{\mathrm{cal}}),
\]
distributed $\chi^2_{T-1}$ under truthful reporting, where $\Var(\hat a^{\mathrm{cal}})$ is the calibration-time variance. Power $\geq 0.95$ at $T\geq 50$ for shrinkage $\epsilon\geq 0.3$.
\end{remark}

Combined with platform-computed attribution, this yields a layered defense: institutional prevention plus statistical detection.

\subsection{Mussa--Rosen screening with private risk aversion}
\label{app:ic_screening}

Suppose risk aversion $\alpha_j$ is private information, drawn from a distribution $G$ on $[\underline\alpha,\bar\alpha]$ with density $g$ satisfying the monotone hazard rate (MHR) condition: $\varphi(\alpha):=\alpha+G(\alpha)/g(\alpha)$ is strictly increasing.

\begin{proposition}[Optimal contract menu under private types]
\label{prop:screening_full}
Under MHR, the platform's revenue-maximizing IC menu offers
\[
\rho_j^{\mathrm{SB}}(\alpha)=\min\!\Big(1,\;\frac{\mathcal{I}_j}{\varphi(\alpha)}\Big),
\]
implementable in dominant strategies without ironing. Per-type deadweight loss from screening (where $\rho^{\mathrm{SB}}<1$) is
\[
\mathrm{DWL}_j(\alpha)=\frac{\alpha\Var(v_j)}{2}\big(\rho_j^*-\rho_j^{\mathrm{SB}}\big)^2=\frac{\mathcal{I}_j^2\Var(v_j)\,(\varphi(\alpha)-\alpha)^2}{2\alpha\,\varphi(\alpha)^2}.
\]
\end{proposition}

\begin{proof}
Single-crossing: $\partial^2 U/\partial\rho\partial\alpha=-\rho\Var(v_j)<0$ for $\rho>0$. Envelope: $(U^*)'(\alpha)=-\rho(\alpha)^2\Var(v_j)/2\leq 0$. Substituting and integrating by parts gives the pointwise problem $\max_\rho[\rho\mathcal{I}_j\Var(v_j)-(\varphi(\alpha)/2)\rho^2\Var(v_j)]$, with FOC $\rho^{\mathrm{SB}}=\mathcal{I}_j/\varphi(\alpha)$. MHR gives $\varphi'>0$, so monotonicity holds without ironing. The cap $\min(1,\cdot)$ enforces $\rho\leq 1$, binding for $\varphi(\alpha)<\mathcal{I}_j$. The DWL formula follows by substituting $\rho^*=\mathcal{I}_j/\alpha$ and $\rho^{\mathrm{SB}}=\mathcal{I}_j/\varphi(\alpha)$.
\end{proof}

\textbf{Separation of signal IC and menu IC.} When the manipulation game is timed sequentially (report type $\to$ receive contract $\to$ observe $\hat a$ $\to$ choose manipulation), the creator's utility Hessian over $(b,\epsilon,\hat\alpha)$ is block-diagonal between the post-contractual signal block $(b,\epsilon)$ and the pre-contractual type block $\hat\alpha$. The platform's design problem therefore decouples: institutional/monitoring controls handle the signal channel (\Cref{app:ic_hierarchy}), and the Mussa--Rosen menu of Proposition~\ref{prop:screening_full} handles the type channel. Under stochastic $\pi$ with $\Cov(\hat a_j\pi,\pi)\neq 0$, off-diagonal coupling reappears and the separation is at most approximate.

\subsection{Dynamic self-policing and competition invariance}
\label{app:ic_dynamic}

Suppose detected Tier~2 bias triggers a transition to inactive status, with continuation-value gap $V^A-V^I>0$ (positive in both thin- and thick-margin regimes by the dynamic results in Appendix~\ref{app:dynamics}). Under constant detection $\Pr(\text{detect}\mid b>0)=\Pr_0$, equilibrium is bang-bang: $b\in\{0,b_{\mathrm{static}}^*\}$, with truthfulness requiring the discount factor to satisfy
\[
\delta\geq\underline\delta_j^{\mathrm{const}}:=\frac{(\rho_j^*\bar\pi)^2}{2(\kappa+\alpha_j(\rho_j^*)^2\sigma_\pi^2)\Pr_0(V^A-V^I)}.
\]
Under proportional detection $\Pr(\text{detect}\mid b)=\phi_d b$ the unique interior bias is $b_{j,\mathrm{dyn}}^*=\rho_j^*\bar\pi/[\kappa+\alpha_j(\rho_j^*)^2\sigma_\pi^2+\delta\phi_d(V^A-V^I)]>0$, strictly decreasing in $\delta$, $\phi_d$, and $V^A-V^I$. Exact truthfulness is unattainable under proportional detection, but the bias is suppressed below the static optimum.

\begin{proposition}[Manipulation incentives are unchanged by competition]
\label{prop:competition_ic}
Under exclusive licensing and Bertrand competition with $K\geq 2$ platforms, the three-tier IC hierarchy of Proposition~\ref{prop:ic_full} is invariant to the number of competitors $K$ and to $\bar U_j$.
\end{proposition}

\begin{proof}
By Theorem~\ref{thm:bertrand} (multi-platform Bertrand equilibrium), $\rho_j^{*(k)}=\mathcal{I}_j^{(k)}/\alpha_j$ is independent of $K$ and $\bar U_j$ because the effective outside option $\hat U_j$ enters platform $k$'s objective additively through $f_j^{(k)}$ and cancels in the $\rho$-derivative. The Tier~1, Tier~2, and Tier~3 manipulation utilities depend only on $(\rho_j^{*(k)},\alpha_j,\Var(\hat a_j),\E[\pi^2],\kappa,\sigma_\pi^2,\Cov(v,\pi))$, none of which depend on $K$ or $\bar U_j$ under exclusive licensing.
\end{proof}

Under non-exclusive licensing the variance of combined per-period payment includes a cross-platform term $2\rho^A\rho^B\Cov(v_j^A,v_j^B)$, which depends on $\rho^B$ and through it on the rival's market structure; IC conditions then become $K$-dependent.

\subsection{Failure modes}
\label{app:ic_failure}

Two boundary cases warrant comment. \emph{Risk-neutral limit:} the cap $\rho_j\leq 1$ binds at $\alpha_j=\mathcal{I}_j$; for $\alpha_j<\mathcal{I}_j$, $\rho_j^*=1$. At $\alpha_j=0$ the Tier~1 variance penalty vanishes, and Tier~1 IC fails weakly when $c=0$. \emph{Pairwise additive pooling:} if creator $i$ pools their attribution report with creator $j$'s, reporting $\hat a_i+\hat a_j$, creator $i$'s payoff variance changes by $(\rho_i^*)^2[\Var(\hat a_j)+2\Cov(\hat a_i,\hat a_j)]$, so additive pooling is unprofitable iff $\Cov(\hat a_i,\hat a_j)>-\Var(\hat a_j)/2$. In music markets, where stylistic similarity drives positive cross-catalog correlation in measured attribution, this condition is typically satisfied and pooling is deterred.

%% file: section_arxiv/D_uniform_welfare.tex
% !TEX root = ../main.tex
\section{Welfare cost of uniform contracts}
\label{app:uniform}

\begin{proof}[Proof of Proposition~\ref{prop:selective_exit}]
We prove each claim in turn.

\textbf{Claim 1 (Selective exit of high-outside-option creators).}
Under a uniform contract $(\bar{f}, \bar{\rho})$, creator $j$ participates if and only if $U_j(\bar{f}, \bar{\rho}) \geq \bar{U}_j$. Expanding:
\begin{equation*}
\bar{f} + \bar{\rho}\,\E[v_{j}] - \frac{\alpha_j}{2}\bar{\rho}^2\,\Var(v_{j}) \geq \bar{U}_j.
\end{equation*}
Rearranging, creator $j$ exits if $\bar{U}_j > \bar{f} + \bar{\rho}\,\E[v_{j}] - \frac{\alpha_j}{2}\bar{\rho}^2\,\Var(v_{j})$. The right-hand side is the utility the uniform contract provides. Creators with high $\bar{U}_j$ (strong standalone careers) are the first for whom this inequality is violated. In the music setting, distinctive creators tend to have both high $\bar{U}_j$ (their original recordings command high streaming revenue and sync fees) and high complementary value in $u$ (their unique styles cannot be replicated by combining others). Hence $\mathcal{S}^{\mathrm{exit}}$ is disproportionately composed of the most valuable creators.

Moreover, when $\bar{\rho}$ is set to the population-average optimal royalty, creators with high $\SNR_j$ (who should receive high $\rho_j^*$ under Theorem~\ref{thm:optimal_rho}) are offered a $\bar{\rho}$ that is \emph{too low} for them---capturing less of their high attributed value---while creators with low $\SNR_j$ are offered a $\bar{\rho}$ that is \emph{too high}---imposing excessive attribution noise risk. The mismatch is worst at the extremes, and the creators most harmed by the mismatch are the distinctive ones whose optimal $\rho_j^*$ is far above $\bar{\rho}$.

\textbf{Claim 2 (Superlinear loss from joint exit).}
For any set $A \subseteq \mathcal{S}$, the loss from removing $A$ simultaneously is:
\begin{equation*}
u(x, r) - u(x \setminus A, r).
\end{equation*}
If $u$ is superadditive across the creators in $A$ (i.e., these creators are complements), then by definition:
\begin{equation*}
u(x, r) - u(x \setminus A, r) \geq \sum_{j \in A}\left[u(x, r) - u(x \setminus \{j\}, r)\right]
\end{equation*}
when $|A| \geq 2$ and the creators in $A$ interact superadditively. The inequality is strict when complementarities are strict. Since $\mathcal{S}^{\mathrm{exit}}$ is concentrated among distinctive creators (Claim 1), and distinctive creators tend to be complements (a reggae producer and a classical composer together enable cross-genre generation that neither supports alone), the loss is superlinear in $|\mathcal{S}^{\mathrm{exit}}|$.

\textbf{Claim 3 (Welfare gap decomposition).}
Under the optimal creator-specific contracts from Theorem~\ref{thm:optimal_rho}, all creators participate (by construction, $\rho_j^*$ and $f_j^*$ are set so that $U_j(\theta_j^*) = \bar{U}_j$), and the platform's welfare is $W^* = \sum_j \E[\pi\, a^*_{j}] - \sum_j \bar{U}_j - \sum_j R_j^*$, where $R_j^* = \frac{\alpha_j}{2}(\rho_j^*)^2 \Var(v_{j})$ is the minimized risk premium.

Under the uniform contract, the platform's welfare has two sources of loss:
\begin{enumerate}
 \item For creators who remain ($j \notin \mathcal{S}^{\mathrm{exit}}$), the risk premium is $R_j^{\mathrm{uniform}} = \frac{\alpha_j}{2}\bar{\rho}^2 \Var(v_{j})$, which is generically larger than $R_j^*$ since $\bar{\rho} \neq \rho_j^*$. The excess is $\mathcal{L}_j^{\mathrm{uniform}} = R_j^{\mathrm{uniform}} - R_j^*$.
 \item For creators who exit ($j \in \mathcal{S}^{\mathrm{exit}}$), the platform loses their entire contribution, which by Claim~2 is superlinear.
\end{enumerate}
Summing these two terms yields the uniform-contract welfare gap.
\end{proof}

%% file: section_arxiv/D2_competition.tex
% !TEX root = ../main.tex
\section{Multi-platform Bertrand competition}
\label{app:competition}

This appendix gives the formal results behind the attribution-moat analysis (Section~\ref{sec:moat}): the proof of the multi-platform Bertrand equilibrium, equilibrium uniqueness, the moat result with both homogeneous and heterogeneous demand, and the welfare comparisons. We use the notation of Section~\ref{sec:setup} extended across $K\geq 2$ platforms: platform $k$ has creator-specific informativeness $\mathcal{I}_j^{(k)}\in(0,1)$ and signal variance $V_j^{(k)}:=\Var(v_{j}^{(k)})$, the multi-platform analogue of $\Var(v_{j})$ from \S\ref{sec:setup} with $v_{j}^{(k)}=\hat a_{j}^{(k)}\pi$.

\subsection{Setting and surplus decomposition}

We assume \textbf{(C1)} exclusive licensing (each creator licenses to at most one platform per period); \textbf{(C2)} heterogeneous informativeness $\mathcal{I}_j^{(k)}$ across platforms; \textbf{(C3)} additive separability of platform $k$'s total surplus across the creators it acquires; and, where indicated, \textbf{(C4)} centered Gaussian noise with $\E[\hat a^{(k)}]=\E[a^*]=0$.

By Theorem~\ref{thm:optimal_rho} applied to platform $k$, the per-creator surplus under royalty $\rho$ is
\[
\Sigma_j^{(k)}(\rho)=\mu_j^{(k)}+\rho\mathcal{I}_j^{(k)} V_j^{(k)}-\frac{\alpha_j}{2}\rho^2 V_j^{(k)},
\]
where $\mu_j^{(k)}$ is the base value from including creator $j$'s catalog. The maximizer is $\rho_j^{*(k)}=\mathcal{I}_j^{(k)}/\alpha_j$ with maximized surplus
\[
\Sigma_j^{*(k)}=\mu_j^{(k)}+\frac{(\mathcal{I}_j^{(k)})^2 V_j^{(k)}}{2\alpha_j}.
\]
Under (C4) and centered attribution, the V-cancellation identity $V_j^{(k)}=\sigma_{a^*}^2\E[\pi^2]/\mathcal{I}_j^{(k)}$ gives the simpler form
\[
\Sigma_j^{*(k)}-\mu_j^{(k)}=\frac{\mathcal{I}_j^{(k)}\sigma_{a^*}^2\E[\pi^2]}{2\alpha_j},
\]
linear in $\mathcal{I}_j^{(k)}$.

\subsection{Bertrand equilibrium}

\begin{theorem}[Multi-platform Bertrand equilibrium]
\label{thm:bertrand}
Under (C1)--(C4), the simultaneous Bertrand--Nash equilibrium for creator $j$ across $K\geq 2$ platforms has the following structure.
\textbf{(i) Royalty.} Each platform offers $\rho_j^{*(k)}=\mathcal{I}_j^{(k)}/\alpha_j$, independent of $K$ and rivals' offers.
\textbf{(ii) Allocation.} Creator $j$ joins $k_j^*=\arg\max_k \Sigma_j^{*(k)}$ if $\Sigma_j^{*(k_j^*)}\geq\bar U_j$, and otherwise exits.
\textbf{(iii) Rents.} The winning platform earns $\Pi_j=\Sigma_j^{*(k_j^*)}-\hat U_j$, where $\hat U_j=\max(\bar U_j,\Sigma_j^{*(2)})$ and $\Sigma_j^{*(2)}$ is the second-highest maximized surplus; the creator earns $U_j^*=\hat U_j$.
\end{theorem}

\begin{proof}[Proof of Theorem~\ref{thm:bertrand}]
\textit{Step 1 (royalty FOC).} Under the binding IR constraint, $f_j^{(k)}=\hat U_j-\rho_j\E[v_j^{(k)}]+(\alpha_j/2)\rho_j^2 V_j^{(k)}$, where $\hat U_j=\max(\bar U_j,\max_{l\neq k}\Sigma_j^{*(l)})$ is the effective outside option. Platform $k$'s profit from creator $j$ is $\Pi_j^{(k)}(\rho)=\Sigma_j^{(k)}(\rho)-\hat U_j$. Since $\hat U_j$ enters as a subtracted constant, the FOC $\partial \Sigma_j^{(k)}/\partial\rho=0$ is independent of $\hat U_j$, $K$, and rival platforms; hence $\rho_j^{*(k)}=\mathcal{I}_j^{(k)}/\alpha_j$ in equilibrium. The SOC is $-\alpha_j V_j^{(k)}<0$, confirming a maximum. Competition affects only the level of rents (through $f_j$), not the slope of incentives (through $\rho_j$).

\textit{Step 2 (Bertrand rent).} The creator joins $k_j^*=\arg\max_k \Sigma_j^{*(k)}$ if at all. To attract $j$, platform $k_j^*$ must offer $U_j\geq\hat U_j=\max(\bar U_j,\Sigma_j^{*(2)})$, where $\Sigma_j^{*(2)}$ is the second-highest surplus. Any $U_j>\hat U_j$ is suboptimal; any $U_j<\hat U_j$ loses the creator. In the Bertrand limit, the winning platform sets $U_j=\hat U_j$ and earns $\Pi_j=\Sigma_j^{*(k^*)}-\hat U_j$.

\textit{Step 3 (allocation).} Creator $j$ joins $k_j^*$ if $\Sigma_j^{*(k^*)}\geq\bar U_j$; otherwise no platform can profitably cover the outside option, and $j$ exits. Equilibrium creator payoff: $U_j^*=\hat U_j=\max(\bar U_j,\Sigma_j^{*(2)})$.
\end{proof}

\begin{proposition}[Equilibrium uniqueness]
\label{prop:bertrand_unique}
The Bertrand equilibrium of Theorem~\ref{thm:bertrand} is unique in pure strategies.
\end{proposition}

\begin{proof}
Exhaustive deviation analysis. (a) The winner cannot increase profit by changing $\rho$ (concavity of $\Sigma_j^{(k)}$ in $\rho$). (b) The winner cannot increase profit by lowering $U_j$ below $\hat U_j$ (loses the creator). (c) A loser cannot profitably attract creator $j$ by offering $U_j>\Sigma_j^{*(l)}$ (negative profit). (d) A loser cannot profitably attract creator $j$ by offering $U_j<\hat U_j$ (creator rejects).
\end{proof}

\subsection{Vertical differentiation rent and the moat}

\begin{corollary}[Rent from being the attribution-precision leader]
\label{cor:vd_rent}
Under (C4) and homogeneous demand ($\mu_j^{(k)}=\mu_j$ for all $k$),
\[
\Pi_j=\frac{(\mathcal{I}_j^{(k^*)}-\mathcal{I}_j^{(2)})\sigma_{a^*}^2\E[\pi^2]}{2\alpha_j}.
\]
\end{corollary}

\begin{proof}
The Bertrand profit is $\Pi_j=\Sigma_j^{*(k^*)}-\Sigma_j^{*(2)}$. Evaluating both maximized surpluses under V-cancellation gives $\Sigma_j^{*(k)}=\mu_j+\mathcal{I}_j^{(k)}\sigma_{a^*}^2\E[\pi^2]/(2\alpha_j)$, and subtraction yields the stated formula.
\end{proof}

We now give the formal version of Proposition~\ref{prop:moat}.

\begin{proof}[Proof of Proposition~\ref{prop:moat}]
\textbf{(i) Dead zone.} Under homogeneous demand and (C4), $\Sigma_j^{*(k)}-\mu_j$ is linear in $\mathcal{I}_j^{(k)}$. Platform~2 becomes the winner only when $\mathcal{I}^{(2)}+\epsilon>\mathcal{I}^{(1)}$, i.e., $\epsilon>\Delta\mathcal{I}_j$. For $\epsilon\leq\Delta\mathcal{I}_j$, Platform~2 remains the runner-up; its surplus improvement raises the creator's effective outside option but generates zero rent for itself.

\textbf{(ii) Post-flip rent.} Direct application of Corollary~\ref{cor:vd_rent} with the updated informativeness $\mathcal{I}^{(2)}+\epsilon$.

\textbf{(iii) Social return.} In the dead zone, $\mathcal{W}=U_j^*+\Pi_j=\Sigma_j^{(1)}$, constant in $\epsilon$: the creator's outside-option gain is exactly offset by the leader's rent reduction. After the flip, $\mathcal{W}=\Sigma_j^{(2)}(\epsilon)=\mu+(\mathcal{I}^{(2)}+\epsilon)\sigma_{a^*}^2\E[\pi^2]/(2\alpha_j)$, strictly increasing.
\end{proof}

\begin{remark}[Heterogeneous demand]
\label{rmk:het_demand}
With heterogeneous demand $\mu^{(k)}$, the flipping threshold becomes
\[
\epsilon^{\mathrm{flip}}=\Delta\mathcal{I}_j+\frac{2\alpha_j(\mu^{(1)}-\mu^{(2)})}{\sigma_{a^*}^2\E[\pi^2]},
\]
extending the dead zone when the leader has a demand advantage and shrinking it when the trailer does.
\end{remark}

\subsection{Welfare comparisons}

\begin{proposition}[Welfare loss under competition]
\label{prop:comp_welfare}
Under Theorem~\ref{thm:bertrand},
\[
\mathcal{L}^{\mathrm{comp}}=\sum_j\mathcal{L}_j^{(k^*_j)},\qquad \mathcal{L}_j^{(k)}=\frac{V_j^{(k)}}{2\alpha_j}(1-\mathcal{I}_j^{(k)})^2.
\]
Under (C4) and homogeneous demand, the competitive allocation minimizes total welfare loss among allocations satisfying (C1)--(C3).
\end{proposition}

\begin{proof}
Under Bertrand, each creator joins $k_j^*=\arg\max_k \Sigma_j^{*(k)}$. Under homogeneous demand, $\Sigma_j^{*(k)}-\mu$ is increasing in $\mathcal{I}_j^{(k)}$, while $\mathcal{L}_j^{(k)}=\sigma_{a^*}^2\E[\pi^2](1-\mathcal{I}_j^{(k)})^2/(2\alpha_j\mathcal{I}_j^{(k)})$ is strictly decreasing in $\mathcal{I}_j^{(k)}$. Therefore competitive allocation assigns each creator to their max-$\mathcal{I}$ platform, which simultaneously minimizes loss; total loss is the sum, by (C3).
\end{proof}

\begin{corollary}[Diminishing returns to entry]
\label{cor:dim_returns}
Under exchangeable informativeness draws, $\E[\Delta W_K]$ is strictly decreasing in $K$. For $\mathcal{I}\sim\mathrm{Uniform}[0,1]$, the marginal gain from adding the $(K{+}1)$-th platform is proportional to $1/((K{+}1)(K{+}2))$.
\end{corollary}

\begin{proof}
Standard order-statistics calculation. $\E[\mathcal{I}_{(K)}]=K/(K+1)$, and $\E[\Delta W_K]\propto\E[\mathcal{I}_{(K+1)}-\mathcal{I}_{(K)}]=1/((K+1)(K+2))$, strictly decreasing in $K$.
\end{proof}

\subsection{Royalty invariance under moral hazard}
\label{app:comp_invariance}

The royalty invariance of Theorem~\ref{thm:bertrand} extends to environments with endogenous quality choice (Appendix~\ref{app:quality}). Under Q1--Q4 and (C1)--(C3), both the static optimal royalty $\rho_j^{*(k)}=\mathcal{I}_j^{(k)}/\alpha_j$ and the equilibrium effort $q_j^*(\rho_j^{*(k)})=\rho_j^{*(k)}\mu_j/D_j(\rho_j^{*(k)})$ depend only on platform $k$'s own primitives $(\mathcal{I}_j^{(k)},\alpha_j,\mu_j,\gamma_j)$, not on $K$ or $\bar U_j$. The effort-channel welfare $W_j(\rho)=\mu_j^2\rho(2-\rho)/(2D_j(\rho))$ inherits the invariance, so the welfare-maximizing royalty $\rho_j^{**}$ is also $K$-invariant. Contract structure (royalty, effort, per-creator welfare) is competition-invariant if and only if the surplus function decomposes as $\Pi_j^{(k)}(\rho)=\Sigma_j^{(k)}(\rho)-\hat U_j$ with $\hat U_j$ not depending on $\rho$.

\subsection{Tension with non-additive utility}

The additive-separability assumption (C3) is the tractable baseline; it conflicts with the non-additive $u_t$ that drives the superlinear loss in the selective-exit analysis. For $u_t=\sum_j s_j+\epsilon\sum_{j<l}c_{jl}$, the Bertrand equilibrium of Theorem~\ref{thm:bertrand} is the $\epsilon=0$ limit, with first-order corrections that may move complementary creators onto the same platform. Quantifying these corrections in general requires solving for equilibrium platform assignments simultaneously and is left to future work; the perturbation argument ensures the headline moat result is robust to weak complementarities.

%% file: section_arxiv/E_modeling.tex
% !TEX root = ../main.tex
\section{Discussion of modeling choices}
\label{app:discussion}

\textbf{Mean-variance vs.\ expected utility.} We adopt mean-variance preferences for tractability and because they are standard in the contract theory literature when the principal designs linear payment schemes. The qualitative results --- that royalty optimality increases in SNR --- hold under general risk-averse expected utility (CARA, CRRA) with appropriate regularity conditions.

\textbf{Independence of attribution noise.} The assumption that $\xi_{j}$ is independent across creators simplifies the variance decomposition but can be relaxed. Correlated attribution errors (e.g., systematic overattribution to popular creators) introduce additional distortions analyzed in ongoing work.

\textbf{Static vs.\ dynamic model.} Our model is static: creators make a one-shot licensing decision. In practice, contracts are renegotiated, and creators observe realized payments before deciding whether to continue. Extending to a dynamic model with learning about SNR and renegotiation is a natural next step.

\textbf{Endogenous attribution cost.} We treat the attribution technology (and hence $\sigma^2_{\xi,j}$) as a platform choice but do not fully solve for the joint optimization over contracts and attribution investment. The first-order condition for optimal investment equates the marginal cost of reducing $\sigma^2_{\xi,j}$ to the marginal welfare gain $|\partial \mathcal{L}_j / \partial \sigma^2_{\xi,j}|$, summed across all creators. This joint optimization is straightforward given Proposition~\ref{prop:welfare_loss} but requires specifying the cost function $C^{\mathrm{attr}}$, which we leave to empirical calibration.

%% file: section_arxiv/F_additional_empirical.tex
% !TEX root = ../main.tex
\section{Additional empirical experiments and the attribution-investment portfolio}
\label{app:experiments}

This appendix collects the additional empirical experiments and the formal water-filling rule that allocates an attribution-improvement budget across creators. The marginal-gain formula from \Cref{prop:welfare_loss} stays in the main text; the optimization that uses it appears here.

\subsection{Detailed empirical setup}
\label{app:setup_details}

We expand on the main-text setup paragraph in \Cref{sec:emp_attribution}.

\textbf{Ground-truth attribution.}
For each creator $j$, the loss-delta ground truth is
\begin{equation}
a^*_{j} \;=\; \ell\!\left(y^{\mathrm{full}}\,\big|\,r;\phi^{-j}\right) \;-\; \ell\!\left(y^{\mathrm{full}}\,\big|\,r;\phi^{\mathrm{full}}\right),
\end{equation}
where $\phi^{\mathrm{full}}$ are the model weights of the full-catalog fine-tune, $\phi^{-j}$ are the weights of the leave-creator-$j$-out fine-tune, $r$ is the prompt, and $y^{\mathrm{full}}$ is the generation produced by $\phi^{\mathrm{full}}$ at $r$. This follows the counterfactual-loss formulation in the data-attribution literature \citep{koh2017understanding,park2023trak,ilyas2022datamodels}, applied at the catalog level rather than the example level. The loss-delta target is model-internal rather than perceptual, which lets every measured $\widehat{\mathcal{I}}_j$ be validated against a quantity the model itself optimizes; perceptual ground truth (e.g., listener-judged stylistic resemblance) is a complementary target left to future work.

\textbf{Catalog selection.}
The $N=22$ artist catalogs are drawn from MTG-Jamendo's autotagging split, with at least ten tracks per artist and at least $40\%$ within-artist genre coherence.
Catalogs span four prolificness bands by track count --- $8$ low ($10$--$15$ tracks), $8$ mid ($16$--$25$), $5$ high ($26$--$50$), and $1$ extra-high ($>50$, capped at $50$) --- and ten primary genres (electronic, classical, hip-hop, pop, rap, latin, alternative, heavy metal, pop-rock, experimental).
The total training set is $475$ tracks; each leave-catalog-out fine-tune drops the corresponding catalog and retrains.

\textbf{Model and fine-tuning.}
We fine-tune Stable Audio Open \citep{evans2024stable} on the curated training set with the VAE frozen, training only the diffusion transformer.
Each of the $N+1=23$ models (one full + $22$ leave-catalog-out) is fine-tuned from the public checkpoint with matched hyperparameters; we use the same prompted-continuation set of $T=200$ evaluation prompts across all models.

\textbf{Attribution-method implementation.}
For both EKFAC influence functions \citep{grosse2023studyinglarge} and D-TRAK \citep{zheng2024dtrak}, we compute per-sample attribution scores against each of the $T$ generated outputs and aggregate to the catalog level by summation.
EKFAC uses block-diagonal Fisher approximation; D-TRAK uses random projections of training and query gradients followed by a kernel-based score.

\textbf{Quantity definitions.}
For each method $m$ and creator $j$ we report four quantities:
$\widehat{\mathcal{I}}_j^{(m)}$ is the regression slope of $a^*$ on $\hat a^{(m)}$ across the $T$ outputs;
$\widehat K_j^{(m)} = \widehat{\Cov}(\hat a^{(m)}, a^*)/\widehat{\Var}(a^*)$ is the multiplicative-bias scale;
$\widehat{\SNR}_j^{(m,\mathrm{cal})}$ is the Gaussian-shortcut SNR computed from the calibrated estimator $\tilde a_{j,t} = \hat a_{j,t}/\widehat K_j^{(m)}$;
and $\widehat{\rho}_j^{(m)}$ is the per-creator Pearson correlation between $\hat a^{(m)}$ and $a^*$ across the $T$ outputs.
Bootstrap $95\%$ confidence intervals are computed by resampling the $T$ outputs with replacement, $1{,}000$ replicates per quantity.

\textbf{Sanity checks.}
The ground truth $a^*$ is positive in $97.6\%$ of $4400$ creator--output pairs (mean $0.0042$, std $0.0059$), confirming that leave-catalog-out retraining reliably worsens loss in the expected direction.
The four-quantity tuple satisfies the identity $\widehat{\mathcal{I}}_j = \widehat{\SNR}^{\mathrm{cal}}_j / \big[\widehat K_j(1+\widehat{\SNR}^{\mathrm{cal}}_j)\big]$ up to numerical precision on the exported data, validating the calibration relation in \Cref{app:assumptions}.

\textbf{Headline statistics.}
\Cref{tab:headline} consolidates the headline measurement numbers referenced throughout \Cref{sec:emp_attribution}.

\begin{table}[t]
\caption{Headline measurement statistics on Stable Audio Open with $N=22$ artist catalogs and $T=200$ generated outputs, under EKFAC and D-TRAK. Each correlation row reports the per-clip cross-creator correlation averaged over clips, with $95\%$ bootstrap CIs in brackets. The $|\widehat K_j|$ row reports the median across the $22$ creators.}
\label{tab:headline}
\centering
\small
\begin{tabular}{lcc}
\toprule
Quantity & EKFAC influence & D-TRAK \\
\midrule
Mean cross-artist Spearman                               & $0.050$ \,$[0.021,\,0.080]$  & $0.054$ \,$[0.022,\,0.087]$ \\
Mean cross-artist Pearson                                & $0.188$ \,$(p<10^{-4})$       & $0.177$ \,$(p<10^{-12})$ \\
Pearson / Spearman ratio                                 & $3.76\times$                  & $3.24\times$ \\
Median $|\widehat K_j|$                                  & $7\,966$                      & $5.25$ \\
Creators with $\widehat K_j<0$ (sign-flipped)            & $6/22$                        & $3/22$ \\
Creators with $\widehat{\SNR}^{\mathrm{cal}}_j>0.1$      & $11/22$                       & $3/22$ \\
Creators with $\widehat{\SNR}^{\mathrm{cal}}_j>1$        & $0/22$                        & $1/22$ \\
Aggregate welfare loss $\mathcal{L}^{(m)}$ at $\alpha=2$ & $1.67\times 10^{-4}$          & $1.69\times 10^{-4}$ \\
Welfare gain over no-attribution                         & $10.6\%$                      & $9.8\%$ \\
\bottomrule
\end{tabular}
\end{table}

\subsection{Per-creator measurement diagnostics}
\label{app:per_creator}

Figure~\ref{fig:results_app1} reports per-creator measurements that aggregate to the headline statistics in the main text Table~\ref{tab:headline}.
At $T=200$ the per-creator bootstrap intervals are wide and span zero for most creators on both methods.
A small subset of creators show statistically significant positive Pearson under both methods, and the per-creator Pearson values are correlated across methods at Spearman $\rho = 0.51$ ($p=0.015$), so methods agree on roughly which creators are detectable even though their cardinal scales differ by orders of magnitude.

\begin{figure*}[t]
\centering
\includegraphics[width=\textwidth]{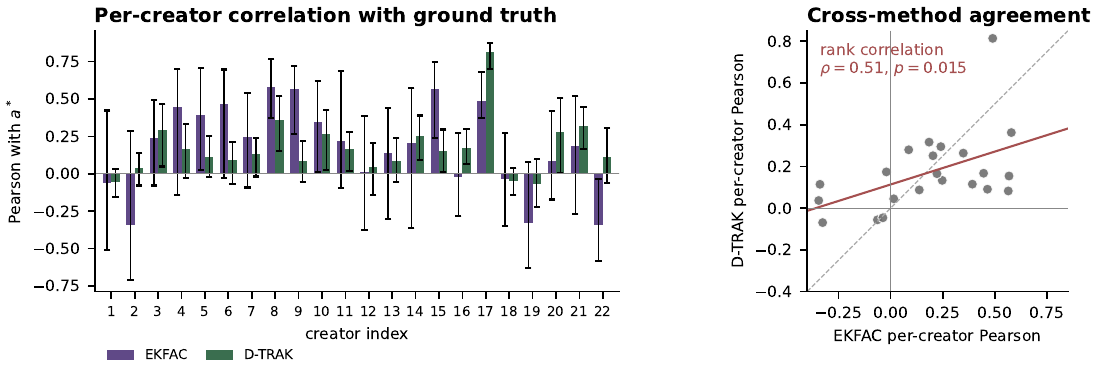}
\caption{Per-creator detail behind the main-text headline numbers. \textbf{(a)}~Per-creator Pearson with $95\%$ bootstrap confidence intervals. Most intervals span zero at $T=200$; a small subset of creators show significant positive correlation under both methods. \textbf{(b)}~Cross-method agreement on per-creator Pearson, with the diagonal in dashed grey and an OLS fit in red. The methods agree at rank correlation $\rho=0.51$ ($p=0.015$) on which creators are detectable, even though their absolute scales differ by three orders of magnitude.}
\label{fig:results_app1}
\end{figure*}

Figure~\ref{fig:results_app2} reports two scaling diagnostics that motivate the framework's marginal-welfare investment rule.
The variance of the loss-delta ground truth grows roughly linearly with catalog size on a log--log scale, with Pearson $r=0.69$ ($p<0.001$).
This scaling is a necessary condition for attribution to discriminate small from large catalogs, since per-creator informativeness is a regression slope on the ground-truth distribution; without dispersion in $a^*$, the slope is undefined.
The right panel overlays the closed-form marginal welfare gain $|\partial\mathcal{L}/\partial\SNR|$ from Proposition~\ref{prop:welfare_loss} (computed at the median $\Var(a^*)$) with per-creator measured points.
Per-creator points scatter above and below the median curve in proportion to each creator's own $\Var(a^*)$.
The shaded region marks the precision range that holds most of the current pool, where the marginal welfare gain from improved attribution is largest --- the same region the water-filling rule below routes the bulk of an attribution-improvement budget into.

\begin{figure*}[t]
\centering
\includegraphics[width=\textwidth]{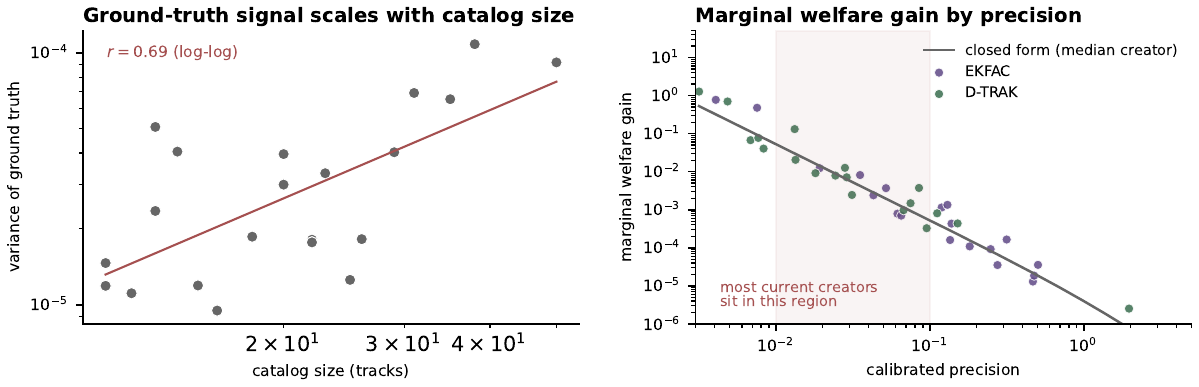}
\caption{Scaling and marginal-welfare diagnostics. \textbf{(a)}~Ground-truth signal variance versus catalog size, log--log axes (Pearson $r=0.69$, $p<0.001$); larger catalogs perturb the model more when removed and are therefore more informative. \textbf{(b)}~Closed-form marginal welfare gain from Proposition~\ref{prop:welfare_loss} (grey curve, evaluated at the median $\Var(a^*)$) overlaid with per-creator measured points. Most of the current creator pool sits in the shaded steep-slope region, where attribution-investment yields the highest welfare return.}
\label{fig:results_app2}
\end{figure*}

\subsection{Optimal attribution-investment portfolio}
\label{app:investment}

Suppose the platform has a budget $B>0$ to invest in improving attribution precision and faces a per-creator improvement technology $g_j:[0,\bar\beta_j]\to\R_+$ that is strictly increasing, strictly concave, $C^1$, with $g_j(0)=0$ and $g_j'(0)<\infty$. Investing $\beta_j$ units of budget in creator $j$ raises their SNR from $S_j$ to $S_j+g_j(\beta_j)$. Adopt additive-Gaussian noise (centered: $\E[\hat a_j]=0$, or $\pi$ deterministic). The platform allocates the budget to maximize total welfare-loss reduction.

\begin{corollary}[Optimal attribution-investment portfolio]
\label{cor:waterfilling}
The welfare-maximizing allocation $\{\beta_j^*\}$ solving
\[
\max_{\beta\geq 0,\sum_j\beta_j=B}\sum_{j=1}^N\big[\mathcal{L}_j(S_j)-\mathcal{L}_j(S_j+g_j(\beta_j))\big]
\]
satisfies the following:
\begin{enumerate}
\item \emph{Existence and uniqueness.} $\beta^*$ is unique by strict concavity of the composed objective.
\item \emph{Priority index.} Define $\mathcal{M}_j:=\phi_j\cdot g_j'(0)$ where
\[
\phi_j=\frac{\Lambda_j(2S_j+1)}{S_j^2(1+S_j)^2},\qquad \Lambda_j=\frac{\E[\pi^2]\Var(a^*_j)}{2\alpha_j}
\]
is the marginal welfare gain from an infinitesimal SNR improvement. Creator $j$ receives investment iff $\mathcal{M}_j$ exceeds the shadow price $\lambda^*$ of the budget constraint.
\item \emph{Closed form under square-root technology.} If $g_j(\beta)=\gamma_j\sqrt{\beta}$, the first-order approximation of the optimal allocation is
\[
\beta_j^{\mathrm{approx}}=\frac{(\phi_j\gamma_j)^2}{\sum_k(\phi_k\gamma_k)^2}\cdot B.
\]
Budget shares follow a square law in marginal welfare gains.
\item \emph{Comparative statics.} $\partial\beta_j/\partial S_j<0$, $\partial\beta_j/\partial\alpha_j<0$, $\partial\beta_j/\partial\Var(a^*_j)>0$.
\end{enumerate}
\end{corollary}

\begin{proof}
\textit{(i)} $\mathcal{L}(S)$ is strictly convex for $S>0$: $\mathcal{L}''=\Lambda(6S^2+6S+2)/[S^3(1+S)^3]>0$. The gain $G_j(s):=\mathcal{L}(S_j)-\mathcal{L}(S_j+s)$ is therefore strictly concave in $s$. Composition of strictly increasing concave $G_j$ with strictly concave $g_j$ is strictly concave, so $h_j(\beta_j):=G_j(g_j(\beta_j))$ is strictly concave on $[0,\bar\beta_j]$. The sum $\sum_j h_j(\beta_j)$ is strictly concave on the compact convex constraint set; existence and uniqueness follow from Weierstrass plus strict concavity.

\textit{(ii)} KKT: $h_j'(\beta_j^*)=\lambda^*$ if $\beta_j^*>0$, else $h_j'(0)\leq\lambda^*$. With $h_j'(0)=G_j'(0)g_j'(0)=\phi_j\cdot g_j'(0)=\mathcal{M}_j$, strict concavity gives the if-and-only-if characterization.

\textit{(iii)} Linearize $G_j(s)\approx\phi_j s$ for small $s$. Then $h_j(\beta_j)\approx\phi_j\gamma_j\sqrt{\beta_j}$. The FOC $\phi_j\gamma_j/(2\sqrt{\beta_j})=\lambda$ yields $\sqrt{\beta_j}=\phi_j\gamma_j/(2\lambda)$; squaring and enforcing the budget constraint gives the stated formula.

\textit{(iv)} Direct differentiation of $\beta_j\propto\phi_j^2$.
\end{proof}

\textbf{Pilot calibration.} At the pilot parameters used throughout the paper ($\alpha_j=2$, $\E[\pi^2]=1$, three creators with $\mathcal{I}\in\{0.617,0.755,0.209\}$), the budget shares are 19.3\% (pop), 5.5\% (reggae), and 75.1\% (classical). The classical creator receives three quarters of the budget despite having the smallest $\Var(a^*)$, because their low SNR places this creator on the steepest part of the welfare-loss curve. This generalizes the main-text observation that the moderate-precision regime has the highest marginal return: when $\Var(\hat a)$ and $\alpha_j$ heterogeneity is taken into account, low-informativeness creators with moderate noise dominate the priority ranking.

\textbf{Approximation guarantee.} The closed-form allocation captures at least fraction $1/(1+\bar\epsilon)$ of the optimal welfare gain under the Prop-1-convention objective, where $\bar\epsilon=\max_j g_j(B)\,/[(S_j+g_j(B))S_j]$. At pilot with $B=0.00751$ (1\% of mean SNR), $\bar\epsilon=0.069$ and the guarantee is $\geq 93.5\%$.

\textbf{Augmented water-filling.} The priority index extends to incorporate quality dividends (from \Cref{app:quality}), retention dividends (from the dynamic results in Appendix~\ref{app:dynamics}), and IC costs (from Proposition~\ref{prop:screening_full}): $\tilde\phi_j=\phi_j+\Delta_j^q+\Delta_j^r-\Delta_j^{\mathrm{IC}}$, with $\tilde\beta_j^{\mathrm{approx}}\propto(\tilde\phi_j\gamma_j)^2$. The augmented allocation concentrates investment more heavily on low-informativeness creators in the pilot calibration.

\subsection{Multi-platform Bertrand competition between attribution methods}
\label{app:market_simulation}

We treat the two state-of-the-art attribution methods evaluated in Part~I --- EKFAC influence functions and D-TRAK --- as two competing platforms in the Bertrand equilibrium of \Cref{thm:bertrand}.
Each platform offers each creator the platform-specific attribution-aligned benchmark contract from \Cref{thm:optimal_rho} computed using its own $\widehat{\mathcal{I}}^{\mathrm{cal},(m)}_j$.
Under exclusive licensing, creator $j$ joins the platform with the higher maximized surplus $\Sigma^{*(m)}_j$, and the trailing platform's private return to attribution improvement is zero up to the leader's lead.

\textbf{Surplus proxy.} We use the leading-order surplus from \Cref{thm:optimal_rho} after the binding-IR substitution,
\[
\Sigma^{*(m)}_j \;\propto\; \big(\widehat{\mathcal{I}}^{\mathrm{cal},(m)}_j\big)^2 \cdot \frac{\Var(a^*_j)}{2\alpha},
\]
which captures the per-creator surplus magnitude up to a normalization absorbed into the dead-zone width.

\textbf{Per-creator equilibrium.}
Of the $22$ creators, EKFAC wins $14$ and D-TRAK wins $8$ (\Cref{fig:results_bertrand}~(a)).
The two methods produce non-trivial divergence on roughly one-third of the population, where the trailing platform's private return is exactly zero --- the dead zone of \Cref{prop:moat}.

\textbf{Dead-zone distribution.}
Dead-zone widths span four orders of magnitude across creators, from $\sim 10^{-11}$ (near-tie) to $\sim 10^{-5}$ (large gap), with median around $10^{-7}$ (\Cref{fig:results_bertrand}~(b)).
Most creators have nearly indifferent surplus across methods --- a small minority drives the population-level dead-zone width that under-investing platforms face.

\begin{figure*}[t]
\centering
\includegraphics[width=\textwidth]{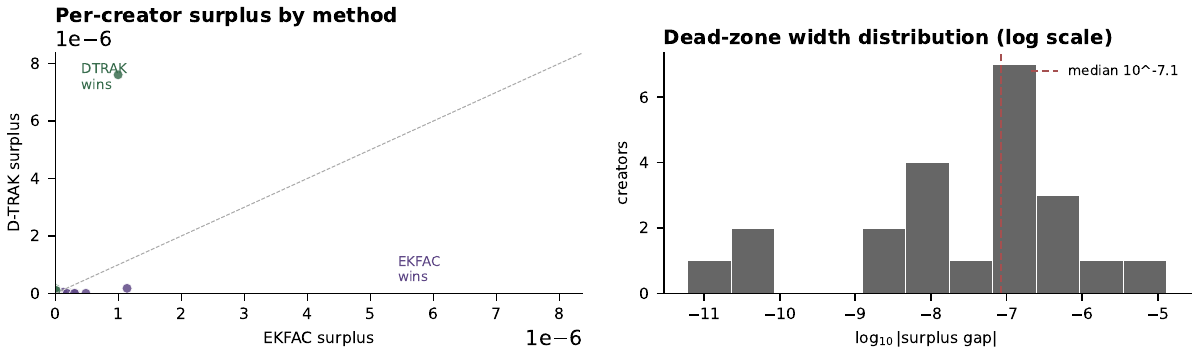}
\caption{Bertrand competition between EKFAC and D-TRAK as two platforms. \textbf{(a)}~Per-creator surplus on each platform; points above the diagonal favour D-TRAK, below favour EKFAC. \textbf{(b)}~Distribution of dead-zone widths (absolute surplus gap) on a logarithmic axis, spanning four orders of magnitude with median around $10^{-7}$.}
\label{fig:results_bertrand}
\end{figure*}

\textbf{Industry-pattern reference.}
\Cref{tab:industry} relates current music-industry payment models to the framework's contract space $\theta_j = (f_j, \rho_j)$.
Existing models occupy the boundary (pure fixed fee or pure royalty); AI training licenses are a setting where the attribution-aligned benchmark can place creators in the interior of the contract space, and the framework's prescription on the measured creator pool (\Cref{sec:emp_market}) gives a concrete instantiation of where those creators would sit under the benchmark.

\begin{table}[t]
\caption{Existing music industry payment models as special cases of the framework's contract space $\theta_j = (f_j, \rho_j)$. Most current models occupy the boundary (pure fixed fee or pure royalty); AI training licenses are a setting where the attribution-aligned benchmark can produce interior contracts.}
\label{tab:industry}
\centering
\small
\begin{tabular}{lllll}
\toprule
Industry model & Payment form & $\theta = (f, \rho)$ & Attribution signal & SNR regime \\
\midrule
Blanket license (ASCAP/BMI)    & Fixed fee & $(f, 0)$    & None              & N/A \\
Pro-rata streaming (Spotify)   & Royalty   & $(0, \rho)$ & Stream counts     & Low (coarse) \\
Sync license (film/TV)         & Fixed fee & $(f, 0)$    & Negotiated        & N/A \\
User-centric (SoundCloud)      & Royalty   & $(0, \rho)$ & Per-user streams  & Medium \\
AI training license (proposed) & Hybrid    & $(f, \rho)$ & Attribution method & Varies \\
\bottomrule
\end{tabular}
\end{table}

%% file: section_arxiv/G_bargaining.tex
% !TEX root = ../main.tex
\section{Bargaining with large rights holders}
\label{app:bargaining}

When a creator is sufficiently irreplaceable (e.g., a major label controlling a unique catalog), the platform cannot credibly offer take-it-or-leave-it terms. Instead, the two parties bargain bilaterally over the contract $\theta_j = (f_j, \rho_j)$. Let $\bar{\Pi}_0$ denote the platform's payoff if negotiation fails (it generates without creator $j$) and $\bar{U}_j$ the creator's outside option. The incremental joint surplus from agreement is
\[
S_j = \E[\pi \cdot a^*_{j}] - \bar{\Pi}_0 - \bar{U}_j.
\]
If $\beta_j \in (0,1)$ is the creator's bargaining power, the Nash bargaining solution splits the surplus so that creator $j$ receives $\bar{U}_j + \beta_j S_j$. The royalty rate $\rho_j$ is still chosen to minimize the risk premium (as in Theorem~\ref{thm:optimal_rho}), but the fixed fee $f_j$ adjusts to implement the bargaining split rather than extracting all surplus above the outside option. The key interaction with attribution precision: when attribution is noisy, both parties prefer to shift toward fixed fees (reducing the noisy component), but the creator's bargaining power determines how the risk-premium savings are divided. A high-$\beta_j$ creator captures most of the efficiency gain from better attribution.

%% file: section_arxiv/H_dynamics.tex
% !TEX root = ../main.tex
\section{Dynamic participation}
\label{app:dynamics}

The static framework of Section~\ref{sec:theory} assumes a one-shot licensing decision. In practice, creators observe realized payments and decide whether to continue licensing. This appendix gives the formal dynamic extension: a Bellman equation for the platform's value from each active creator, the steady-state participation rate, and a two-regime characterization of when the platform optimally pays a retention rent.

\subsection{Setup}

\textbf{(D1) Myopic exit rule.} At the end of period $t$, creator $j$ exits iff $w_{j,t}<\bar U_j$. This captures liquidity-constrained creators or behavioral anchoring on $\bar U_j$.

\textbf{(D2) Re-entry.} An exited creator returns the next period with probability $\lambda_j\in(0,1)$.

\textbf{(D3) Stationary contract.} The platform commits to $\theta_j=(f_j^*,\rho_j^*)$ from Theorem~\ref{thm:optimal_rho} each period, plus an optional retention rent $s_j\geq 0$ added to $f_j$.

\textbf{(D4) Discount factor.} The platform discounts at $\delta\in(0,1)$.

\textbf{(D5) Gaussian payments.} $w_{j,t}\mid x_{j,t}=1\sim\mathcal{N}(\E[w_j],(\rho_j^*)^2\sigma_{v,j}^2)$.

Define net surplus $N_j:=\E[a_{j,t}^*\pi_t]-\bar U_j>0$ and net margin $\mu_j^{\mathrm{net}}:=N_j-\mathcal{L}_j>0$ (creator profitable net of alignment loss). For brevity we write $\sigma_{v,j}^2:=\Var(v_{j,t})$ throughout this appendix.

\subsection{Static exit probability and the $\alpha$-cancellation}

\begin{theorem}[Dynamic participation]
\label{thm:dynamic}
Under (D1)--(D5) and Theorem~\ref{thm:optimal_rho}'s payment rule:

\textbf{(a) Static exit probability.} With zero retention rent,
\[
p_j^{\mathrm{static}}=\Phi(-d_j),\qquad d_j=\frac{\mathcal{I}_j\sigma_{v,j}}{2}.
\]
The exit buffer $d_j$ is independent of risk aversion $\alpha_j$: $\rho_j^*\propto 1/\alpha_j$ exactly offsets the risk-aversion scaling in the variance penalty.

\textbf{(b) Bellman equation.} The platform's value from an active creator with retention rent $s\geq 0$ is
\[
V_j^A(s)=\frac{\Pi_j(s)}{1-\delta\tilde\sigma_j(s)},
\]
where $\Pi_j(s)=\mu_j^{\mathrm{net}}-s$, $\sigma_j(s)=1-\Phi(-d_j-s/(\rho_j^*\sigma_{v,j}))$ is the retention probability, $\tilde\sigma_j(s)=\omega_j+(1-\omega_j)\sigma_j(s)$ is the effective continuation probability, and $\omega_j=\delta\lambda_j/(1-\delta(1-\lambda_j))$ is the re-entry discount factor.

\textbf{(c) Steady-state participation.} In the ergodic steady state,
\[
n_j^*=\frac{\lambda_j}{\lambda_j+p_j(\delta)}.
\]
\end{theorem}

\begin{proof}
\textit{(a)} Under binding IR, $\E[w_j]=\bar U_j+R_j^*$ with $R_j^*=(\alpha_j/2)(\rho_j^*)^2\sigma_{v,j}^2=\mathcal{I}_j^2\sigma_{v,j}^2/(2\alpha_j)$. Payment standard deviation is $\rho_j^*\sigma_{v,j}=\mathcal{I}_j\sigma_{v,j}/\alpha_j$. Exit iff $w_{j,t}<\bar U_j$, so $p_j=\Phi(-R_j^*/(\rho_j^*\sigma_{v,j}))$. The ratio simplifies: $d_j=R_j^*/(\rho_j^*\sigma_{v,j})=\mathcal{I}_j\sigma_{v,j}/2$, manifestly independent of $\alpha_j$.

\textit{(b)} Bellman: $V^A(s)=\Pi(s)+\delta[\sigma(s)V^A(s)+(1-\sigma(s))V^I]$, $V^I=\delta[\lambda V^A(s)+(1-\lambda)V^I]$. Solving the second for $V^I=\omega V^A$ and substituting into the first yields $V^A(s)[1-\delta(\sigma+(1-\sigma)\omega)]=\Pi(s)$, i.e., $V^A(s)=\Pi(s)/[1-\delta\tilde\sigma(s)]$.

\textit{(c)} Flow balance of the two-state Markov chain: $p_j n_j=\lambda_j(1-n_j)$, giving $n_j=\lambda_j/(\lambda_j+p_j)$.
\end{proof}

\subsection{Two-regime characterization}

\begin{proposition}[Two-regime characterization]
\label{prop:two_regime}
The platform optimally sets $s_j^*>0$ (thick-margin regime) if and only if
\[
\mu_j^{\mathrm{net}}>\mu_j^{\mathrm{crit}}:=\frac{p_j^{\mathrm{static}}+\omega_j}{\sigma_j'(0)}.
\]
Otherwise $s_j^*=0$ (thin-margin regime) and exit occurs at the static rate $p_j^{\mathrm{static}}$.
\end{proposition}

\begin{proof}
Define $\Gamma(s)=F(s)-G(s)$ with $F=\delta(1-\omega)\Pi\sigma'$ and $G=1-\delta\tilde\sigma$. The first-order condition for the Bellman is $\Gamma(s^*)=0$, and $(F-G)'=\delta(1-\omega)\Pi\sigma''<0$ on the domain by log-concavity of the normal CDF. Strictly decreasing $\Gamma$ implies at most one zero. Evaluate $\Gamma(0)=\delta(1-\omega)\mu^{\mathrm{net}}\sigma'(0)-(1-\delta\tilde\sigma(0))$; the FOC has an interior root iff $\Gamma(0)>0$, which rearranges to $\mu^{\mathrm{net}}>\mu^{\mathrm{crit}}$. Otherwise the corner solution $s^*=0$ is optimal.
\end{proof}

\textbf{Patient-limit behavior.} In the thin-margin regime, $p_j(\delta)=p_j^{\mathrm{static}}$ for every $\delta$: patience is irrelevant when margins are thin, and the exit rate is fully determined by the static buffer. In the thick-margin regime, $s_j^*(\delta)$ is strictly increasing in $\delta$ with a finite patient limit $s_j^\infty=\lim_{\delta\to 1}s_j^*(\delta)$ satisfying $\Pi_j(s^\infty)\sigma_j'(s^\infty)=\lambda_j(1-\sigma_j(s^\infty))$, and the patient-limit exit probability $p_j^\infty=1-\sigma_j(s^\infty)<p_j^{\mathrm{static}}$. The retention rent paid by a patient platform is bounded above by the net margin.

\textbf{Comparative statics.} $\partial p_j^{\mathrm{static}}/\partial\mathcal{I}_j=-(\sigma_{v,j}/2)\phi(d_j)<0$ (higher informativeness reduces exit) and $\partial p_j^{\mathrm{static}}/\partial\sigma_{v,j}=-(\mathcal{I}_j/2)\phi(d_j)<0$. The latter is initially counterintuitive---higher signal variance \emph{reduces} exit---and reflects the fact that the risk premium grows quadratically in $\sigma_v$ while the payment standard deviation grows only linearly, so the buffer-to-std-dev ratio $d_j$ scales with $\sigma_v$.

\textbf{Retention dividend.} Improving $\mathcal{I}_j$ raises the steady-state participation rate and yields a marginal welfare gain
\[
\Delta_j^{\mathrm{ret}}=\frac{\lambda_j(N_j-\mathcal{L}_j)\sigma_{v,j}\phi(d_j)}{2(\lambda_j+p_j)^2}.
\]
This dividend feeds the augmented water-filling rule of Appendix~\ref{app:investment}.

\textbf{Music-specific stakes.} In the music context, the dynamic margin matters because creator exit destroys not just supply volume but \emph{stylistic diversity}: the loss of a distinctive reggae producer or classical composer cannot be compensated by adding more generic pop producers, and the steady-state participation rate of the most distinctive long-tail catalogs is the operative variable for the platform's long-run generation quality. Combined with the static results in the selective-exit analysis, the dynamic framework predicts that uniform compensation drives a persistent participation gap concentrated on the most valuable cohort.

%% file: section_arxiv/I_synthetic.tex
% !TEX root = ../main.tex
\section{Synthetic data value chains}
\label{app:synthetic}

The platform may use its generative model not only to serve end users but also to produce \emph{synthetic music data} for downstream purposes: fine-tuning specialized models, licensing to third parties, or reducing dependence on future human-created music licensing. This creates a multi-stage value chain: the platform licenses original music $\to$ trains a model $\to$ generates synthetic data $\to$ extracts downstream value.

In stage~1, the platform generates output $y$ with attributed value $v_{j}$ for creator $j$. In stage~2, the platform uses its model to produce synthetic data $\tilde{D}$ and extracts downstream value $V^{\mathrm{syn}}(\tilde{D})$. The total value attributable to creator $j$'s original data is not just $v_{j}$ but $v_{j} + \Delta v_j^{\mathrm{syn}}$, where $\Delta v_j^{\mathrm{syn}}$ captures $j$'s marginal contribution to the synthetic data's downstream value.

A payment rule based only on $v_{j}$ systematically underpays creators whose data is valuable for synthetic generation. This has three implications for the framework:
\begin{itemize}
 \item \emph{Amplified displacement.} Synthetic data generated in a creator's style continues competing in the market without further licensing, compounding the displacement effect $\delta_j$ over time.
 \item \emph{Compounding attribution noise.} Measuring $\Delta v_j^{\mathrm{syn}}$ requires attributing value through a chain (original data $\to$ model $\to$ synthetic data $\to$ downstream value), and the noise at each stage compounds. The effective SNR for the full-chain attribution is lower than the single-stage SNR, strengthening the case for fixed-fee components that do not depend on per-output measurement.
 \item \emph{Payment rule implications.} If the platform extracts significant synthetic-data value without compensating original creators, pure per-output royalties understate total value extraction. This strengthens the case for hybrid payments that include an upfront component reflecting the data's option value for future synthetic generation.
\end{itemize}

%% file: refs.bib
@article{choe2024your,
  title={What is your data worth to gpt? llm-scale data valuation with influence functions},
  author={Choe, Sang Keun and Ahn, Hwijeen and Bae, Juhan and Zhao, Kewen and Kang, Minsoo and Chung, Youngseog and Pratapa, Adithya and Neiswanger, Willie and Strubell, Emma and Mitamura, Teruko and others},
  journal={arXiv preprint arXiv:2405.13954},
  year={2024}
}

@article{thickstun2023anticipatory,
  title={Anticipatory music transformer},
  author={Thickstun, John and Hall, David and Donahue, Chris and Liang, Percy},
  journal={TMLR},
  year={2024}
}

@inproceedings{donahue2022melody,
  title={Melody transcription via generative pre-training},
  author={Donahue, Chris and Thickstun, John and Liang, Percy},
  booktitle={ISMIR},
  year={2022}
}

@article{acemoglu2022toomuch,
  author    = {D. Acemoglu and A. Makhdoumi and A. Malekian and A. Ozdaglar},
  title     = {Too much data: Prices and inefficiencies in data markets},
  journal   = {American Economic Journal: Microeconomics},
  volume    = {14},
  number    = {4},
  pages     = {218--256},
  year      = {2022}
}

@inproceedings{agarwal2019marketplace,
  author    = {A. Agarwal and M. Dahleh and T. Sarkar},
  title     = {A marketplace for data: An algorithmic solution},
  booktitle = {Proceedings of the 2019 ACM Conference on Economics and Computation},
  pages     = {701--726},
  year      = {2019}
}

@misc{barnett2024embeddings,
  author    = {J. Barnett and H. Flores Garcia and B. Pardo},
  title     = {Exploring musical roots: Applying audio embeddings to empower influence attribution for a generative music model},
  year      = {2024},
  note      = {arXiv preprint arXiv:2401.14542}
}

@inproceedings{choi2025unlearning,
  author    = {W. Choi and J. Koo and K. W. Cheuk and J. Serr{\`a} and M. A. Martinez-Ramirez and Y. Ikemiya and N. Murata and Y. Takida and W.-H. Liao and Y. Mitsufuji},
  title     = {Large-scale training data attribution for music generative models via unlearning},
  booktitle = {Advances in Neural Information Processing Systems (NeurIPS), Creative AI Track},
  year      = {2025}
}

@misc{agostinelli2023musiclm,
  author    = {A. Agostinelli and T. I. Denk and Z. Borsos and J. Engel and M. Verzetti and A. Caber and N. Zeghidour and C. Frank},
  title     = {{MusicLM}: Generating music from text},
  year      = {2023},
  note      = {arXiv preprint arXiv:2301.11325}
}

@misc{deng2023copyright,
  author    = {J. Deng and X. Jiang and S. Zhang and S. Zhang and H. Lakkaraju and R. Gao and C. Donahue and J. Ma},
  title     = {Computational copyright: Towards a royalty model for music generative {AI}},
  year      = {2023},
  note      = {arXiv preprint arXiv:2312.06646}
}

@inproceedings{ghorbani2019data,
  author    = {A. Ghorbani and J. Zou},
  title     = {Data {S}hapley: Equitable valuation of data for machine learning},
  booktitle = {Proceedings of the 36th International Conference on Machine Learning},
  pages     = {2242--2251},
  year      = {2019}
}

@misc{georgiev2023journeytrak,
  author    = {K. Georgiev and J. Vendrow and H. Salman and S. M. Park and A. Madry},
  title     = {The journey, not the destination: How data guides diffusion models},
  year      = {2023},
  note      = {arXiv preprint arXiv:2312.06205}
}

@inproceedings{giordano2019swiss,
  author    = {R. Giordano and W. Stephenson and R. Liu and M. Jordan and T. Broderick},
  title     = {A {S}wiss army infinitesimal jackknife},
  booktitle = {Proceedings of the 22nd International Conference on Artificial Intelligence and Statistics},
  pages     = {1139--1147},
  year      = {2019}
}

@article{holmstrom1979moral,
  author    = {B. Holmstr{\"o}m},
  title     = {Moral hazard and observability},
  journal   = {The Bell Journal of Economics},
  volume    = {10},
  number    = {1},
  pages     = {74--91},
  year      = {1979}
}

@article{holmstrom1987aggregation,
  author    = {B. Holmstr{\"o}m and P. Milgrom},
  title     = {Aggregation and linearity in the provision of intertemporal incentives},
  journal   = {Econometrica},
  volume    = {55},
  number    = {2},
  pages     = {303--328},
  year      = {1987}
}

@inproceedings{jia2019towards,
  author    = {R. Jia and D. Dao and B. Wang and F. A. Hubis and N. Hynes and N. M. G{\"u}rel and B. Li and C. Zhang and D. Song and C. Spanos},
  title     = {Towards efficient data valuation based on the {S}hapley value},
  booktitle = {Proceedings of the 22nd International Conference on Artificial Intelligence and Statistics},
  pages     = {1167--1176},
  year      = {2019}
}

@inproceedings{ilyas2022datamodels,
  author    = {A. Ilyas and S. M. Park and L. Engstrom and G. Leclerc and A. Madry},
  title     = {Datamodels: Predicting predictions from training data},
  booktitle = {Proceedings of the 39th International Conference on Machine Learning},
  pages     = {9525--9587},
  year      = {2022}
}

@article{jones2020nonrivalry,
  author    = {C. I. Jones and C. Tonetti},
  title     = {Nonrivalry and the economics of data},
  journal   = {American Economic Review},
  volume    = {110},
  number    = {9},
  pages     = {2819--2858},
  year      = {2020}
}

@inproceedings{koh2017understanding,
  author    = {P. W. Koh and P. Liang},
  title     = {Understanding black-box predictions via influence functions},
  booktitle = {Proceedings of the 34th International Conference on Machine Learning},
  pages     = {1885--1894},
  year      = {2017}
}

@article{markowitz1952portfolio,
  author    = {H. Markowitz},
  title     = {Portfolio selection},
  journal   = {The Journal of Finance},
  volume    = {7},
  number    = {1},
  pages     = {77--91},
  year      = {1952}
}

@article{maskin1984monopoly,
  author    = {E. Maskin and J. Riley},
  title     = {Monopoly with incomplete information},
  journal   = {The RAND Journal of Economics},
  volume    = {15},
  number    = {2},
  pages     = {171--196},
  year      = {1984}
}

@article{mussa1978monopoly,
  author    = {M. Mussa and S. Rosen},
  title     = {Monopoly and product quality},
  journal   = {Journal of Economic Theory},
  volume    = {18},
  number    = {2},
  pages     = {301--317},
  year      = {1978}
}

@inproceedings{park2023trak,
  author    = {S. M. Park and K. Georgiev and A. Ilyas and G. Leclerc and A. Madry},
  title     = {{TRAK}: Attributing model behavior at scale},
  booktitle = {Proceedings of the 40th International Conference on Machine Learning},
  pages     = {27074--27113},
  year      = {2023}
}

@article{pratt1964risk,
  author    = {J. W. Pratt},
  title     = {Risk aversion in the small and in the large},
  journal   = {Econometrica},
  volume    = {32},
  number    = {1/2},
  pages     = {122--136},
  year      = {1964}
}

@techreport{riaa2024lawsuit,
  author      = {{Recording Industry Association of America (RIAA)}},
  title       = {Major record labels file copyright infringement suits against {AI} music generators {S}uno and {U}dio},
  year        = {2024},
  institution = {Recording Industry Association of America (RIAA)},
  type        = {Press release},
  note        = {Press release, June 2024}
}

@techreport{euaiact2024,
  author      = {{European Parliament and Council of the European Union}},
  title       = {Regulation ({EU}) 2024/1689 laying down harmonised rules on artificial intelligence ({AI} Act)},
  year        = {2024},
  institution = {European Parliament and Council of the European Union},
  note        = {Official Journal of the European Union}
}

@techreport{ifpi2025,
  author    = {{International Federation of the Phonographic Industry (IFPI)}},
  title     = {Global Music Report 2025},
  institution = {IFPI},
  year      = {2025}
}

@book{scharf1991statistical,
  author    = {L. L. Scharf},
  title     = {Statistical Signal Processing: Detection, Estimation, and Time Series Analysis},
  publisher = {Addison-Wesley},
  year      = {1991}
}

@article{shannon1948mathematical,
  author    = {C. E. Shannon},
  title     = {A mathematical theory of communication},
  journal   = {The Bell System Technical Journal},
  volume    = {27},
  number    = {3},
  pages     = {379--423},
  year      = {1948}
}

@inproceedings{serra2010audio,
  author    = {J. Serr{\`a} and E. G{\'o}mez and P. Herrera},
  title     = {Audio cover song identification and similarity: Background, approaches, evaluation, and beyond},
  booktitle = {Advances in Music Information Retrieval},
  pages     = {307--332},
  publisher = {Springer},
  year      = {2010}
}

@inproceedings{spijkervet2021clmr,
  author    = {J. Spijkervet and J. A. Burgoyne},
  title     = {Contrastive learning of musical representations},
  booktitle = {Proceedings of the 22nd International Society for Music Information Retrieval Conference (ISMIR)},
  year      = {2021}
}

@inproceedings{wang2003shazam,
  author    = {A. L.-C. Wang},
  title     = {An industrial-strength audio search algorithm},
  booktitle = {Proceedings of the 4th International Conference on Music Information Retrieval (ISMIR)},
  pages     = {7--13},
  year      = {2003}
}

@inproceedings{zhang2025fairshare,
  author    = {L. Zhang and C. Jiao and B. Li and C. Xiong},
  title     = {Fairshare data pricing via data valuation for large language models},
  booktitle = {Advances in Neural Information Processing Systems (NeurIPS)},
  year      = {2025}
}

@article{bergemann2018design,
  author    = {D. Bergemann and A. Bonatti and A. Smolin},
  title     = {The design and price of information},
  journal   = {American Economic Review},
  volume    = {108},
  number    = {1},
  pages     = {1--48},
  year      = {2018}
}

@inproceedings{chen2022selling,
  author    = {Y. Chen and H. Liu and A. Karbasi},
  title     = {Selling data to a machine learner},
  booktitle = {International Conference on Machine Learning (ICML)},
  year      = {2022}
}

@inproceedings{guruganesh2021contracts,
  author    = {G. Guruganesh and J. Schneider and J. R. Wang},
  title     = {Contracts under Moral Hazard and Adverse Selection},
  booktitle = {Proceedings of the 22nd ACM Conference on Economics and Computation (EC)},
  year      = {2021}
}

@techreport{cong2022data,
  author    = {L. W. Cong and W. Wei and D. Xie and L. Zhang},
  title     = {Data, AI, and economic growth},
  institution = {NBER},
  type      = {NBER Working Paper},
  year      = {2022}
}

@book{cover2006elements,
  author    = {T. M. Cover and J. A. Thomas},
  title     = {Elements of Information Theory},
  publisher = {Wiley-Interscience},
  edition   = {2nd},
  year      = {2006}
}

@inproceedings{pruthi2020estimating,
  author    = {G. Pruthi and F. Liu and S. Kale and M. Sundararajan},
  title     = {Estimating training data influence by tracing gradient descent},
  booktitle = {Advances in Neural Information Processing Systems (NeurIPS)},
  year      = {2020}
}

@inproceedings{kwon2024datainf,
  author    = {Y. Kwon and E. Wu and K. Wu and J. Zou},
  title     = {{DataInf}: Efficiently estimating data influence in {LoRA}-tuned {LLMs} and diffusion models},
  booktitle = {International Conference on Learning Representations (ICLR)},
  year      = {2024}
}

@misc{grosse2023studying,
  author    = {R. Grosse and J. Bae and C. Anil and others},
  title     = {Studying large language model generalization with influence functions},
  year      = {2023},
  note      = {arXiv preprint arXiv:2308.03296}
}

@inproceedings{engstrom2024dsdm,
  author    = {L. Engstrom and A. Feldmann and A. M\k{a}dry},
  title     = {{DsDm}: Model-aware dataset selection with datamodels},
  booktitle = {International Conference on Machine Learning (ICML)},
  year      = {2024}
}

@inproceedings{bae2024training,
  author    = {J. Bae and W. Lin and J. Lorraine and R. Grosse},
  title     = {Training data attribution via approximate unrolled differentiation},
  booktitle = {Advances in Neural Information Processing Systems (NeurIPS)},
  year      = {2024}
}

@inproceedings{copet2023simple,
  author    = {J. Copet and F. Kreuk and I. Gat and T. Remez and D. Kant and G. Synnaeve and Y. Adi and A. D{\'e}fossez},
  title     = {Simple and controllable music generation},
  booktitle = {Advances in Neural Information Processing Systems (NeurIPS)},
  year      = {2023}
}

@misc{forsgren2022riffusion,
  author    = {S. Forsgren and H. Martiros},
  title     = {{Riffusion}: Stable diffusion for real-time music generation},
  year      = {2022},
  url       = {https://riffusion.com/about}
}

@misc{evans2024stable,
  author    = {Z. Evans and J. D. Parker and C. J. Carr and Z. Zukowski and J. Taylor and J. Pons},
  title     = {Stable {A}udio {O}pen},
  year      = {2024},
  note      = {arXiv preprint arXiv:2407.14358}
}

@inproceedings{bogdanov2019mtg,
  author    = {D. Bogdanov and M. Won and P. Tovstogan and A. Porter and X. Serra},
  title     = {The {MTG}-{J}amendo dataset for automatic music tagging},
  booktitle = {Machine Learning for Music Discovery Workshop, International Conference on Machine Learning (ICML)},
  year      = {2019}
}

@misc{grosse2023studyinglarge,
  author    = {R. Grosse and J. Bae and C. Anil and N. Elhage and A. Tamkin and A. Tajdini and B. Steiner and D. Li and E. Durmus and E. Perez and E. Hubinger and K. Lukosiute and K. Nguyen and N. Joseph and S. McCandlish and J. Kaplan and S. R. Bowman},
  title     = {Studying large language model generalization with influence functions},
  year      = {2023},
  note      = {arXiv preprint arXiv:2308.03296}
}

@inproceedings{zheng2024dtrak,
  author    = {X. Zheng and T. Pang and C. Du and Q. Liu and J. Jiang and M. Lin},
  title     = {Intriguing properties of data attribution on diffusion models},
  booktitle = {International Conference on Learning Representations (ICLR)},
  year      = {2024}
}

@inproceedings{mlodozeniec2025influence,
  author    = {Bruno Kacper Mlodozeniec and Runa Eschenhagen and Juhan Bae and Alexander Immer and David Krueger and Richard E. Turner},
  title     = {Influence functions for scalable data attribution in diffusion models},
  booktitle = {International Conference on Learning Representations (ICLR)},
  year      = {2025},
  url       = {https://openreview.net/forum?id=esYrEndGsr}
}

@inproceedings{vandenoord2017vqvae,
  author    = {A. van den Oord and O. Vinyals and K. Kavukcuoglu},
  title     = {Neural discrete representation learning},
  booktitle = {Advances in Neural Information Processing Systems (NeurIPS)},
  year      = {2017}
}

@inproceedings{dieleman2018challenge,
  author    = {S. Dieleman and A. van den Oord and K. Simonyan},
  title     = {The challenge of realistic music generation: modelling raw audio at scale},
  booktitle = {Advances in Neural Information Processing Systems (NeurIPS)},
  year      = {2018}
}

@misc{dhariwal2020jukebox,
  author    = {P. Dhariwal and H. Jun and C. Payne and J. W. Kim and A. Radford and I. Sutskever},
  title     = {{Jukebox}: A generative model for music},
  year      = {2020},
  note      = {arXiv preprint arXiv:2005.00341}
}

@inproceedings{ho2020ddpm,
  author    = {J. Ho and A. Jain and P. Abbeel},
  title     = {Denoising diffusion probabilistic models},
  booktitle = {Advances in Neural Information Processing Systems (NeurIPS)},
  year      = {2020}
}

@inproceedings{rombach2022highresolution,
  author    = {R. Rombach and A. Blattmann and D. Lorenz and P. Esser and B. Ommer},
  title     = {High-resolution image synthesis with latent diffusion models},
  booktitle = {Proceedings of the IEEE/CVF Conference on Computer Vision and Pattern Recognition (CVPR)},
  year      = {2022}
}

@inproceedings{caillon2025livemusic,
  author    = {A. Caillon and B. McWilliams and C. Tarakajian and I. Simon and I. Manco and J. Engel and N. Constant and Y. Li and others},
  title     = {Live music models},
  booktitle = {Advances in Neural Information Processing Systems (NeurIPS), Creative AI Track},
  year      = {2025},
  note      = {arXiv preprint arXiv:2508.04651}
}

@misc{gong2026acestep,
  author    = {J. Gong and Y. Song and W. Zhao and S. Wang and S. Xu and J. Guo and X. Yang},
  title     = {{ACE-Step 1.5}: Pushing the boundaries of open-source music generation},
  year      = {2026},
  note      = {arXiv preprint arXiv:2602.00744}
}

@misc{silberling2026suno,
  author       = {A. Silberling},
  title        = {{AI} music generator {S}uno hits {2M} paid subscribers and \$300{M} in annual recurring revenue},
  howpublished = {TechCrunch},
  year         = {2026},
  note         = {Published February 27, 2026},
  url          = {https://techcrunch.com/2026/02/27/ai-music-generator-suno-hits-2-million-paid-subscribers-and-300m-in-annual-recurring-revenue/}
}

@misc{deezer2026aiuploads,
  author       = {{Deezer}},
  title        = {{Deezer}: {AI}-generated tracks now represent 44\% of all new uploaded music},
  howpublished = {Deezer Newsroom},
  year         = {2026},
  note         = {Published April 20, 2026},
  url          = {https://newsroom.deezer.com/2026/04/ai-generated-tracks-represent-44-of-new-uploaded-music/}
}

@article{akerlof1970lemons,
  author    = {G. A. Akerlof},
  title     = {The market for ``lemons'': Quality uncertainty and the market mechanism},
  journal   = {The Quarterly Journal of Economics},
  volume    = {84},
  number    = {3},
  pages     = {488--500},
  year      = {1970}
}

@article{hirshleifer1971private,
  author    = {J. Hirshleifer},
  title     = {The private and social value of information and the reward to inventive activity},
  journal   = {The American Economic Review},
  volume    = {61},
  number    = {4},
  pages     = {561--574},
  year      = {1971}
}

@article{donahue2025hookpad,
  title={Hookpad Aria: A copilot for songwriters},
  author={Donahue, Chris and Wu, Shih-Lun and Kim, Yewon and Carlton, Dave and Miyakawa, Ryan and Thickstun, John},
  journal={arXiv preprint arXiv:2502.08122},
  year={2025}
}

@inproceedings{kim2025amuse,
  title={Amuse: Human-AI collaborative songwriting with multimodal inspirations},
  author={Kim, Yewon and Lee, Sung-Ju and Donahue, Chris},
  booktitle={Proceedings of the 2025 CHI conference on human factors in computing systems},
  pages={1--28},
  year={2025}
}
